\documentclass[10pt,journal]{IEEEtran}

\usepackage{amssymb}
\usepackage{amsmath}
\usepackage{cite}
\usepackage{url}
\usepackage{xcolor}
\usepackage{cite,graphicx,amsmath,amssymb}
\usepackage{fancyhdr}
\usepackage{mdwmath}
\usepackage{mdwtab}
\usepackage{caption}
\usepackage{amsthm}
\usepackage{setspace}
\usepackage{hyperref}
\usepackage{algorithm}
\usepackage{algorithmic}
\usepackage{multirow}
\usepackage{makecell}
\usepackage{mathtools}
\usepackage{subcaption}
\usepackage{bm}
\usepackage{tikz}
\usepackage{xurl}
\usepackage{stfloats}
\usepackage{mathrsfs}

\usepackage{booktabs}
\usepackage{multirow}
\usepackage{siunitx}
\usepackage{balance}

\hypersetup{colorlinks=true,
linkcolor=blue,
citecolor=blue,      
urlcolor=black,
}

\graphicspath{{./src/img/}}

\newtheorem{theorem}{Theorem}

\newtheorem{lemma}{Lemma}

\newtheorem{corollary}{Corollary}

\allowdisplaybreaks
\NewDocumentCommand{\multiubrace}{mmm}
 {% #1 = items, #2 = descriptions, #3 = separators
  \egreg_multiubrace:nnn {#1} {#2} {#3}
 }

\expandafter\def\expandafter\normalsize\expandafter{%
	\normalsize%
	\setlength\abovedisplayskip{4pt}%
	\setlength\belowdisplayskip{4pt}%
	\setlength\abovedisplayshortskip{2pt}%
	\setlength\belowdisplayshortskip{2pt}%
}
\begin{document}
\title{Tri-Hybrid Beamforming Design for Fully-Connected Pinching Antenna Systems} 
\author{{Cheng-Jie Zhao, Zhaolin Wang,~\IEEEmembership{Member, IEEE},\\ Hyundong Shin,~\IEEEmembership{Fellow, IEEE}, and Yuanwei Liu,~\IEEEmembership{Fellow, IEEE}}

\thanks{Cheng-Jie Zhao, Zhaolin Wang, and Yuanwei Liu are with the Department of Electrical and Electronic Engineering, The University of Hong Kong, Hong Kong (e-mail: chengjie\_zhao@connect.hku.hk; zhaolin.wang@hku.hk; yuanwei@hku.hk). \\
\indent Hyundong Shin is with the Department of Electronics and Information Convergence Engineering, Kyung Hee University, 1732 Deogyeong-daero, Giheung-gu, Yongin-si, Gyeonggi-do 17104, Republic of Korea (e-mail: hshin@khu.ac.kr).}}
\maketitle

\begin{abstract}
	A novel fully-connected (FC) tri-hybrid beamforming (THB) architecture is proposed for pinching antenna systems (PASS). In contrast to conventional sub-connected (SC) PASS, the proposed FC architecture employs a tunable phase-shifter network to interconnect all radio frequency (RF) chains with all waveguides. This facilitates a THB framework that integrates conventional hybrid analog-digital beamforming with pinching beamforming. A weighted sum-rate (WSR) optimization problem is then formulated to jointly optimize the transmit beamformers and pinching antenna (PA) positions. Two algorithms are developed to address this challenging non-convex problem. 1) \textit{Fractional programming (FP)-based algorithm:} This algorithm directly maximizes the WSR using an FP-based alternating optimization framework. Particularly, a success-history based adaptive differential evolution (SHADE) method is proposed to optimize PA positions, effectively addressing the intractable multimodal objective function. 2) \textit{Zero-forcing (ZF)-based algorithm:} To reduce design complexity, zero-forcing is employed for transmit beamforming. The PA positions are subsequently optimized to maximize the WSR via a modified SHADE method. Simulation results validate the effectiveness of the proposed algorithms, revealing that the FC-THB PASS achieves WSR comparable to the SC architecture while delivering superior energy efficiency with fewer RF chains.
\end{abstract}
\section{Introduction}
In recent decades, the landscape of wireless communication has been profoundly transformed by the emergence of multiple-input multiple-output (MIMO) technology~\cite{MIMO}. By exploiting spatial degrees of freedom, MIMO systems have dramatically improved both spectral efficiency and link reliability, laying the groundwork for modern multi-user communication networks. As wireless services continue to expand toward higher data rates, wider coverage, and more heterogeneous application scenarios, the need for large-scale antenna arrays has become increasingly prominent. Nevertheless, scaling the antenna infrastructure to meet these demands inevitably leads to excessive hardware costs, higher energy consumption, and significant computational complexity, all of which pose critical challenges to practical deployment. \\
\indent To address the limitations of conventional fixed-structure antennas, reconfigurable antenna (RA) technologies have emerged, introducing a new level of flexibility in electromagnetic control~\cite{RA}. By adaptively adjusting their polarization, operating frequency, or radiation pattern, RAs enable energy-efficient beam steering without proportional increases in radio frequency (RF) hardware. Representative examples include intelligent reflecting surfaces (IRS) and fluid/movable antennas (FA/MA). IRS manipulates wireless propagation by programming the reflection phase of passive elements, thereby enhancing signal strength and suppressing interference~\cite{IRS, IRS1}. In contrast, FA and MA achieve transceiver-side reconfigurability. Particularly, FA typically adjusts the position of antennas through the controlled movement of conductive liquids within dielectric structures, whereas MA mainly advocates mechanical displacement of antenna elements. By dynamically adapting to channel variations, these architectures can exploit small-scale fading and enhance spatial diversity, offering a promising means of improving link reliability and throughput~\cite{FAS, MAS}. 

\subsection{Prior Works} 
Despite these advances, the reconfigurability of IRS, FA, and MA remains fundamentally constrained. FA and MA typically operate within limited spatial apertures, e.g., on the order of only several wavelengths, thus offering insufficient capability to mitigate large-scale path loss. Meanwhile, the performance of IRS strongly depends on accurate channel estimation and phase alignment, which poses significant challenges for practical implementation and real-time adaptation. As a remedy, a novel reconfigurable antenna concept, termed the pinching-antenna system (PASS), has recently been proposed~\cite{PASS1}. PASS enables continuous and wide-range spatial control of electromagnetic apertures by adjusting the excitation and pinching points along dielectric waveguides~\cite{PASS2}. Through such controllable design, the effective radiation position and pattern can be dynamically reconfigured over extended regions, thereby facilitating the establishment of line-of-sight (LoS) links and substantially improving transmission quality. Furthermore, constructed from low-cost dielectric waveguides and detachable pinching elements, PASS provides a lightweight, scalable, and cost-effective means of achieving adaptive electromagnetic manipulation. Its hardware simplicity allows flexible deployment and easy reuse without the need for complex RF front-end circuitry. \\
\indent Owing to its promising advantages, PASS has attracted growing research interest in recent years. To list a few, a comprehensive analytical framework for PASS was developed in~\cite{PASS3}, deriving closed-form expressions for outage probability and average rate under a fundamental single-user scenario. The study revealed that PASS outperforms traditional antenna systems in terms of both reliability and data rate, making it a viable candidate for next-generation wireless technology. The array gain achieved by multiple pinching antennas (PAs) on a waveguide was analyzed in~\cite{PASS4}, uncovering the optimal number of antennas and their spatial spacing to maximize beamforming gain. In what follows, optimal beamforming design and PA position optimization strategies are investigated to maximize the sum-rate in downlink and uplink communication scenarios~\cite{PASS5, PASS6, PASS7}. \\
\indent As a further advance, a physics-based hardware model for PASS, along with an in-depth analysis of the electromagnetic field behavior, was presented in~\cite{PASS8}, providing insights into the fundamental physical principles governing the performance of PASS. Based on this model, a transmit power minimization problem was formulated, addressing the joint optimization of transmit and pinching beamforming under both continuous and discrete PA activations. Given the practical challenges in implementing continuous PA deployment, dual-timescale system designs and deployment protocols were proposed in~\cite{PASS9} and~\cite{PASS10} to enhance energy efficiency in real-world applications. In addition, recent advancements have explored the use of large language models in beamforming design~\cite{PASS11}, highlighting the growing trend of leveraging machine learning techniques to optimize antenna systems and further improve their performance. 

\subsection{Motivation and Contributions}
Although current studies on PASS have demonstrated its superior advantages in establishing reliable LoS links, enhancing transmission quality, and improving achievable data rates, most of these works have been limited to sub-connected (SC) architectures, where each dielectric waveguide is fed by a dedicated RF chain. This structure indeed simplifies the hardware implementation and facilitates independent waveguide control, making it suitable for small-scale or proof-of-concept prototypes. However, such a one-to-one mapping between RF chains and waveguides imposes a fundamental constraint on scalability. As the number of waveguides increases to achieve higher beamforming gain or finer spatial resolution, the number of required RF chains must grow proportionally, leading to prohibitively high hardware cost, circuit complexity, and energy consumption. This contradicts the design goals of PASS, namely achieving reconfigurability and adaptability through simple and cost-efficient hardware. \\
\indent Furthermore, the SC architecture inherently limits signal diversity and spatial flexibility. Since each waveguide is connected to a single RF chain, inter-waveguide cooperation is impossible, preventing the joint utilization of spatial degrees of freedom (DoFs) across the entire aperture. Consequently, the system cannot fully exploit the rich spatial diversity offered by multiple radiation points along the dielectric waveguides. To address these limitations, this paper aims to investigate the novel fully-connected (FC) architecture for PASS and explore joint beamforming and position optimization for the proposed system. The key contributions of this work are summarized as follows:
\begin{itemize}
	\item We propose an FC architecture for PASS, which employs a tunable phase-shifter network to connect RF chains and waveguides. This architecture naturally supports a tri-hybrid beamforming (THB) framework that integrates digital, analog, and pinching beamforming. Allowing each waveguide to be jointly fed by every RF chain, FC PASS offers flexible configuration in the number of RF chains and beamforming design DoFs to reduce hardware and computational cost. 
	\item We formulate a joint transmit beamforming and PA position optimization problem to maximize the weighted sum-rate (WSR) of the proposed FC-THB PASS subject to transmit power and inter-PA distance constraints. To address this highly non-convex and coupled problem, we first show that the transmit power constraint can be equivalently removed. Then, a fractional programming (FP)-based alternating optimization (AO) algorithm is developed, where each variable is iteratively updated while keeping the others fixed. In particular, a success-history based adaptive differential evolution (SHADE) method is proposed for optimizing the PA positions, addressing the intractable multi-modal optimization challenges.
	\item We further propose a zero-forcing (ZF)–based algorithm as a low-complexity alternative, given that the FP-based algorithm requires multiple iterative updates involving all designed variables for convergence. In this approach, the PA positions are first optimized via the proposed SHADE to maximize the WSR achieved by ZF beamforming, after which the digital and analog beamformers are obtained through decomposition of the corresponding ZF beamformer.
	\item We provide comprehensive numerical results to validate the proposed FC-THB PASS and the effectiveness of the developed algorithms. The results demonstrate that: 1) The FC architecture, optimized with the proposed AO algorithm, achieves comparable WSR performance to the SC counterpart using significantly fewer RF chains. For instance, it delivers $90\%$ of the SC performance with only half the RF chains; and 2) The proposed low-complexity ZF algorithm offers an attractive trade-off between performance and computational complexity, making it well-suited for energy- and cost-constrained deployments.
\end{itemize}

\subsection{Organization and Notations}
The remainder of this paper is structured as follows. Section~\ref{sec:model} introduces the system model and formulates the WSR maximization problem for the proposed FC-THB PASS. Section~\ref{sec3} details the proposed AO algorithm for addressing the formulated problem. Section~\ref{sec4} provides a ZF beamforming-based algorithm as a low-complexity sub-optimal alternative for the formulated problem. Numerical results evaluating the performance of the proposed FC architecture and algorithms under various system configurations are presented in Section~\ref{sec5}. Finally, Section~\ref{sec6} concludes the paper.

\emph{Notations:} Scalars, vectors/matrices, and Euclidean subspaces are denoted by regular, boldface, and calligraphic letters, respectively. The sets of complex, real, and integer numbers are represented by $\mathbb{C}$, $\mathbb{R}$, and $\mathbb{Z}$, respectively. The inverse, transpose, conjugate transpose (Hermitian transpose), and trace operations are represented by $(\cdot)^{-1}$, $(\cdot)^T$, $(\cdot)^H$, and $\mathrm{Tr}(\cdot)$, respectively. The absolute value and Euclidean norm are indicated by $|\cdot|$ and $\|\cdot\|$, respectively. The expectation operator is denoted by $\mathbb{E}{\cdot}$. The real part of a complex number is denoted by $\Re {\cdot}$. An identity matrix of size $N \times N$ is denoted by $\mathbf{I}_N$. The vector ${{\bf e}_k}$ denotes the $k$-th canonical basis vector, whose entries are all zeros except for a one in the $k$-th position. Accordingly, its self-outer product ${{\bf E}_k} = {{\bf e}_k}{{\bf e}_k}^H$ represents a diagonal matrix with a single one at the $(k,k)$-th entry and zeros elsewhere. $\odot$ denotes element-wise multiplication Big-O notation is represented by $O(\cdot)$.

\section{System Model and Problem Formulation} \label{sec:model}
%\begin{figure*}[tb]
%	\centering
%	\subfloat[Tri-hybrid Beamforming Schematic]{
%		\label{sys1}
%		\includegraphics[scale=0.45]{Tri-hybrid Schematic}
%	} \hfill
%	\subfloat[MIMO-PASS Serving Scenario]{
%		\label{sys2}
%		\includegraphics[scale=0.28]{PASS System}
%	}
%	\caption{Illustration of the investigated Tri-hybrid MIMO-PASS}
%	\label{system_model}
%\end{figure*} 

\begin{figure}[tb]
	\centering
	\includegraphics[scale=0.2]{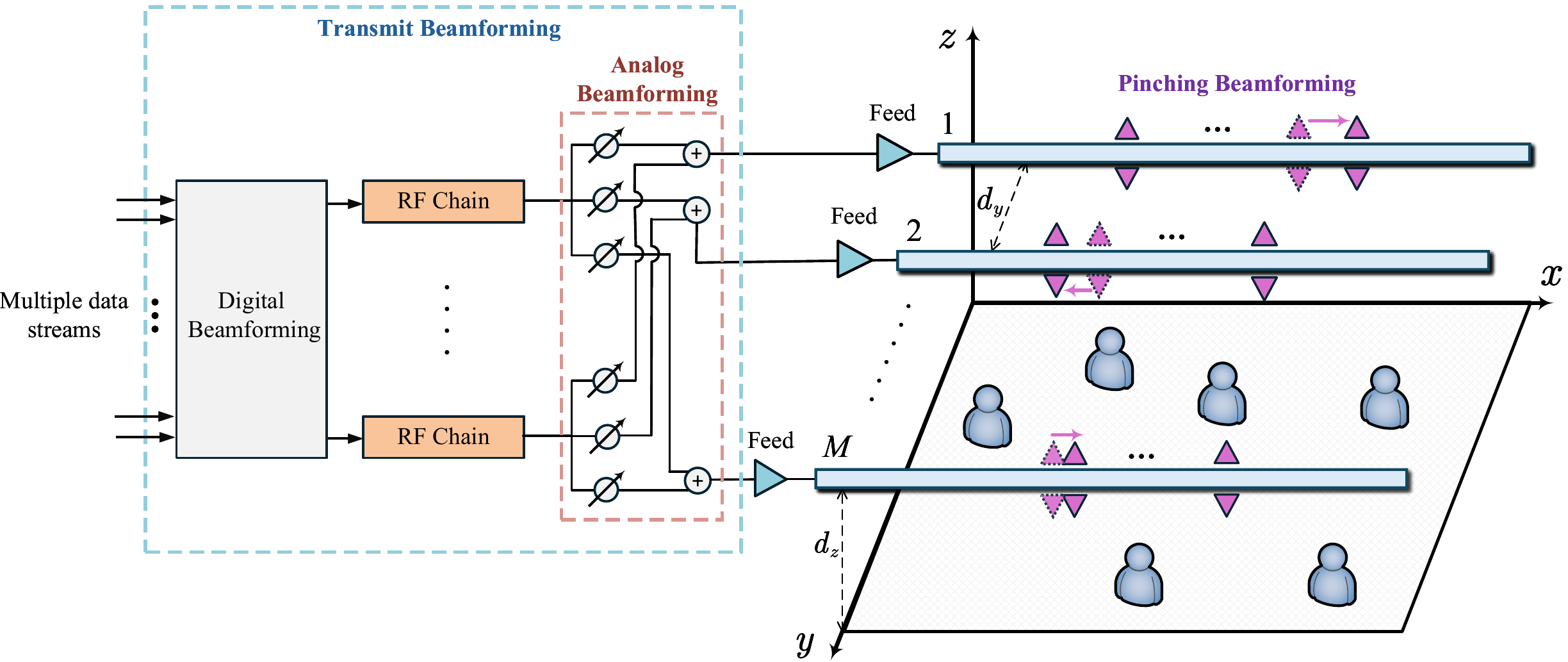}
	\caption{Illustration of the investigated FC-THB PASS}
	\label{system_model}
\end{figure} 

As shown in Fig.~\ref{system_model}, an FC-THB PASS is considered. Specifically, $M$ dielectric waveguides are installed parallelly to the $x$-axis at height $d_z$ and spaced by $d_y$ along the $y$-axis, with each waveguide having length $L$ and equipped with $N$ PAs. These PAs can freely move across the waveguide and extract signals from the waveguide at the pinched locations. $K$ single-antenna users are assumed to randomly locate at a $D_x \times D_y \times d_z$ space and simultaneously served by $N \times M$ PAs. The set of waveguides, PAs in each waveguide, and users are denoted as $\mathcal{M} = \left\{ {1 \le m \le M\left| {m \in \mathbb{Z}} \right.} \right\}$, ${\mathcal{N}_m} = \left\{ {1 \le n \le N\left| {m \in \mathcal{M}} \right.} \right\}$, and $\mathcal{K} = \left\{ {1 \le k \le K\left| {k \in \mathbb{Z}} \right.} \right\}$, respectively. On the $m$-th waveguide, the position of the $n$-th PA is represented as ${\boldsymbol{\psi} _{m,n}^p} = [{{x_{m,n}},\left( {m - 1} \right){d_y},{d_z}}]^T$, where $x_{m,n} \in [0, L]$ is the location of the $n$-th PA on the $m$-th PA. The feed point's location for this waveguide is represented as ${\boldsymbol{\psi} _{m,0}^p} = \left[ {0,\left( {m - 1} \right){d_y},{d_z}} \right]^T$ while the position of the $k$-th user is expressed as ${\boldsymbol{\psi} _k^u} = \left[ {x_k^u,y_k^u,z_k^u} \right]$. 

\subsection{Signal Model}
In the proposed FC-THB architecture, data streams ${\bf{s}} \in \mathbb{C}^{K}$ intended for $K$ users are firstly multiplexed and digitally precoded by the digital beamformer ${{\bf{W}}_{{\rm{BB}}}} \in \mathbb{C}^{N_{{\rm{RF}}} \times K}$, where $N_{{\rm{RF}}}$ is the number of RF chains. The precoded signals are then processed by the analog beamformer ${{\bf{W}}_{{\rm{RF}}}} \in \mathbb{C}^{M \times N_{{\rm{RF}}}}$ and subsequently fed into the dielectric waveguides. Accordingly, the input signal to the waveguides can be expressed as 
\begin{equation}
	{\bf{r}} = {{\bf{W}}_{{\rm{RF}}}}{{\bf{W}}_{{\rm{BB}}}}{\bf{s}}.
\end{equation}
Unlike the SC architecture, where each RF chain exclusively drives a single waveguide, waveguides in the proposed FC architecture receive signals mixed from all RF chains and enable more flexible connectivity and data flow. \\
\indent Each PA radiates a phase-shifted version of the signal propagating within its waveguide, creating a beamforming DoF dependent on the PA locations, which refers to as pinching beamforming. We adopt a base station (BS)-centric design, implying the BS determines the PA positions to carry out pinching beamforming and optimize overall system performance. Additionally, an equal power model is considered, assuming all PAs on a single waveguide radiate identical power. Therefore, the signal propagation response for the $n$-th PA on the $m$-th waveguide is expressed as \cite{PASS8}:
\begin{align} \label{W_PB_L}
	{g_{m,n}} &= \frac{1}{{\sqrt N }} \exp \left\{ {-j\frac{{2\pi }}{{{\lambda _g}}}{\left\| {\psi _{m,n}^p - \psi _{m,0}^p} \right\|_2}} \right\} \notag \\
	&= \frac{1}{{\sqrt N }} \exp \left\{ {-j\frac{{2\pi }}{{{\lambda _g}}} {{x_{m,n}}}} \right\}, 
\end{align}
where ${\lambda _g} = \frac{\lambda }{{{n_{{\rm{eff}}}}}}$ is the guided wavelength with $\lambda$ and ${{{n_{{\rm{eff}}}}}}$ representing the free-space wavelength and effective refractive index of a dielectric waveguide, respectively~\cite{neff}. The propagation response vector for the $m$-th waveguide is then represented as ${{\bf{g}}_m}\left( {{{\bf{x}}_m}} \right) = {\left[ {{g_{m,1}},{g_{m,2}}, \cdots ,{g_{m,N}}} \right]^T},$ which is a function of the PA location vector on this waveguide ${{\bf{x}}_m} = \left[ {{x_{m,1}}, \cdots ,{x_{m,N}}} \right]^T$. Collecting the propagation response vectors, the pinching beamformer can be represented by a block diagonal matrix
\begin{equation} 
	{{\bf{W}}_{{\rm{PB}}}} (\mathbf{X}) = \left[ {\begin{array}{*{20}{c}}
			{{{\bf{g}}_1}}&{\bf{0}}& \cdots &{\bf{0}}\\
			{\bf{0}}&{{{\bf{g}}_2}}& \cdots &{\bf{0}}\\
			\vdots & \vdots & \ddots & \vdots \\
			{\bf{0}}&{\bf{0}}& \cdots &{{{\bf{g}}_M}}
	\end{array}} \right] \in {\mathbb{C}^{MN \times M}},
\end{equation} 
where $\mathbf{X} = [{{\bf{x}}_1}, \cdots, {{\bf{x}}_M}] \in \mathbb{C}^{N \times M}$ is the overall PA location matrix. Consequently, the signal transmitted by the FC-THB PASS is given by 
\begin{equation}
	\mathbf{z} =  {{\bf{W}}_{{\rm{PB}}}}(\mathbf{X}) {{\bf{W}}_{{\rm{RF}}}}{{\bf{W}}_{{\rm{BB}}}}{\bf{s}} \in {\mathbb{C}^{MN}}.
\end{equation}

\subsection{Channel Model}
In this paper, we consider only LoS path in the free space. This assumption is justified for two main reasons. First, dielectric waveguides generally operate at high-frequency bands, such as the millimeter-wave range~\cite{PASS12}, where the LoS component can be up to 20 dB stronger than the non-LoS (NLoS) components~\cite{PASS13}. Second, in our considered high-frequency scenarios, even when NLoS propagation exists, its path loss exponent is typically larger than that of the LoS path, leading to significantly weaker received power. Consequently, the NLoS components contribute negligibly to the overall channel response and are thus omitted from the analysis. Under this assumption, the channel between the $k$-th user and the PAs deployed on the $m$-th waveguide can be expressed as
\begin{equation} \label{H_L}
	{\bar {{\bf{h}}}_{k,m}} = {\left[ {\frac{{\eta \exp \left\{ {j{\textstyle{{2\pi } \over \lambda }}{D_{k,m,1}}} \right\}}}{{{D_{k,m,1}}}}, \cdots ,\frac{{\eta \exp \left\{ {j{\textstyle{{2\pi } \over \lambda }}{D_{k,m,N}}} \right\}}}{{{D_{k,m,N}}}}} \right]^T}
\end{equation}
where $\eta$ is the coefficient proportional to the effective surface of the pinching elements and ${D_{k,m,n}}$ is the distance between the $n$-th PA on the $m$-th waveguide and the $k$-th user, given by 
\begin{align}
	&{D_{k,m,n}} = {\left\| {\psi _{m,n}^p - \psi _k^u} \right\|} \notag \\
	&= \sqrt {{{\left( {{x_{m,n}} - x_k^u} \right)}^2} + {{\left( {\left( {m - 1} \right){d_y} - y_k^u} \right)}^2} + (d_z-z_k^u)^2}.
\end{align}
\indent Stacking the channel vectors related to all waveguide, the effective channel response vector for the $k$-th user yields 
\begin{equation}
	{{\bf{h}}_k}(\mathbf{X}) = {\left[ {{{\bar {\bf{h}}}_{k,1}^T},{{\bar {\bf{h}}}_{k,2}^T}, \cdots ,{{\bar {\bf{h}}}_{k,M}^T}} \right]^T} \in \mathbb{C}^{MN},
\end{equation}
which is also a function of the PA locations. Consequently, the signal received by the $k$-th user can be represented as
\begin{equation} \label{signal_model}
	y_k = {{\bf{h}}_k^H}(\mathbf{X}) {\bf{z}} + n_k = {{\bf{h}}_k^H}(\mathbf{X}) {{\bf{W}}_{{\rm{PB}}}}(\mathbf{X}){{\bf{W}}_{{\rm{RF}}}}{{\bf{W}}_{{\rm{BB}}}}{\bf{s}} + n_k,
\end{equation}
where $n_k$ is the additive white Gaussian noise suffered by the $k$-th user following the circularly symmetric complex Gaussian (CSCG) distribution with zero mean and variance $\sigma_k^2$, i.e., $n_k \sim {\cal CN} (0, \sigma_k^2)$. The overall downlink signal for all users is given by
\begin{equation}
	{\bf{y}} = {\bf{H}}^H(\mathbf{X}){{\bf{W}}_{{\rm{PB}}}}(\mathbf{X}){{\bf{W}}_{{\rm{RF}}}}{{\bf{W}}_{{\rm{BB}}}}{\bf{s}} + {\bf{n}},
\end{equation}
where ${\bf{y}} = {\left[ {{y_1}, \cdots ,{y_K}} \right]^T} \in \mathbb{C}^{K}$, ${\bf{n}} = {\left[ {{n_1}, \cdots ,{n_K}} \right]^T} \in \mathbb{C}^{K}$, and ${\bf{H}} = {\left[ {{{\bf{h}}_1},{{\bf{h}}_2}, \cdots ,{{\bf{h}}_K}} \right]} \in \mathbb{C}^{MN \times K}$. In this paper, the source with $\mathbb{E} \left\{ {{\bf{s}}{{\bf{s}}^H}} \right\} = {{\bf{I}}_K}$ and $\mathbb{E} \left\{ {{\bf{s}}{{\bf{n}}^H}} \right\} = {{\bf{0}}_K}$ is assumed. Additionally, the channel state information (CSI), mainly including the coordinate of user or distance between user and PAs, is considered known to the BS through some efficient CSI acquisition or sensing strategies.

\subsection{Problem Formulation}
In this section, the WSR metric is considered to guide the optimization of designed variables. Specifically, transmit beamformers and PAs' positions are jointly optimized to maximize the weighted sum of rates achieved by all users. Based on~(\ref{signal_model}), the signal-to-interference-plus-noise (SINR) for decoding the desired signal at user $k$ is given by
\begin{align}
	{\gamma _k} &= \frac{{{{\left| {{\bf{h}}_k^H\left( {\bf{X}} \right){{\bf{w}}_k}} \right|}^2}}}{{\sum\nolimits_{i \ne k} {{{\left| {{\bf{h}}_k^H\left( {\bf{X}} \right){{\bf{w}}_i}} \right|}^2}}  + \sigma _k^2}} \notag \\
	&= \frac{{{\rm{Tr}}( {{{\bf{H}}_k}{\bf{W}}{{\bf{E}}_k}{{\bf{W}}^H}} )}}{{{\rm{Tr}}( {{{\bf{H}}_k}{\bf{W}}\left( {{{\bf{I}}_K} - {{\bf{E}}_k}} \right){{\bf{W}}^H}} ) + \sigma _k^2}}, 
\end{align}
where ${{\bf{w}}_k}$ is the $k$-th column of the composite beamformer ${\bf{W}}={{\bf{W}}_{{\rm{PB}}}}(\mathbf{X}){{\bf{W}}_{{\rm{RF}}}}{{\bf{W}}_{{\rm{BB}}}}$, ${{\bf{H}}_k} = {{\bf{h}}_k}{\bf{h}}_k^H$, and ${{\bf{E}}_k} = {{\bf{e}}_k}{\bf{e}}_k^T$. With a serving priority $\left\{ {{\beta _k}} \right\}_{k = 1}^K$, which satisfies $\sum\nolimits_{k = 1}^K {{\beta _k}}=1$, the WSR is represented as
\begin{equation} \label{rate}
	R = \sum\limits_{k = 1}^K {{\beta _k}{{\log }_2}\left( {1 + {\gamma _k}} \right)}.
\end{equation}
Consequently, the WSR maximization problem can be formulated as
\begin{subequations}\label{1}
	\begin{align} 
		\left( {\rm{P1}}\right)  \  &\mathop {\max }\limits_{{{\bf{W}}_{{\rm{BB}}}},{{\bf{W}}_{{\rm{RF}}}},{\bf{X}}} \quad R  \\
		\label{1a} \mathrm{s.t.} \ \  &\left\| {{{\bf{W}}_{{\rm{RF}}}}{{\bf{W}}_{{\rm{BB}}}}} \right\|_F^2 = P, \\
		\label{1aa} &{\left| {{{\bf{W}}_{{\rm{RF}}}}\left( {i,j} \right)} \right|^2} = 1,\forall i,j,   \\
		\label{1b} &0 \le {x_{m,n}} \le L, n \in \mathcal{N}_m, m \in \mathcal{M},  \\
		\label{1c} &{x_{m,n}} - {x_{m,n - 1}} \ge \Delta d, , n \in \mathcal{N}_m \backslash \left\{ 1 \right\}, m \in \mathcal{M}. 
	\end{align}
\end{subequations}
\indent In (P1), constraint (\ref{1a}) and (\ref{1aa}) regulate the transmitting power of the system and analog beamforming, respectively. Constraints (\ref{1b}) and (\ref{1c}) restrict the adjustable PA locations, where $\Delta d > 0$ defines the minimum inter-PA distance to avoid mutual coupling. More particularly, constraint (\ref{1c}) is based on the assumption that PA locations in each vector $\mathbf{x}_m$ are sorted in non-decreasing order, such that PAs with larger indices are placed further from the feed point of the waveguide. \\
\indent Problem (P1) is extremely challenging to solve, as it entails the joint optimization of three coupled beamforming stages. The optimization variables are intricately intertwined within the logarithmic, complex-exponential, and fractional terms of the objective function (\ref{rate}). Moreover, the transmit power constraint (\ref{1a}) introduces a bilinear coupling between the two transmit beamformers, further complicating the problem structure. Consequently, both the objective function (\ref{1}) and the constraints (\ref{1a})–(\ref{1aa}) are non-convex with respect to the joint optimization variables. To address these challenges, we proposed two optimization algorithms in the following sections, namely a \emph{FP-based algorithm} and a \emph{ZF-based algorithm}.

\section{Tri-Hybrid Beamforming via Fractional Programming} \label{sec3}
%In detail, the box constraint (\ref{1b}) and minimum inter-PA distance constraint (\ref{1c}) can be recast as
%\begin{equation} \label{1bb}
%	{{\bf{0}}_{N \times M}} \preceq {\bf{X}} \preceq L \cdot {{\bf{1}}_{N \times M}}
%\end{equation}
%and
%\begin{equation} \label{1cc}
%	{\bf{DX}} \succeq \Delta d \cdot {{\bf{1}}_{\left( {N - 1} \right) \times M}},
%\end{equation}
%where ${\bf{D}}$ is defined as a differential operator $${\bf{D}} = \left[ {\begin{array}{*{20}{c}}
%		{ - 1}&1&0& \cdots &0\\
%		0&{ - 1}&1& \cdots &0\\
%		\vdots & \vdots & \ddots & \ddots & \vdots \\
%		0& \cdots &0&{ - 1}&1
%\end{array}} \right] \in {\left\{ { - 1,0,1} \right\}^{\left( {N - 1} \right) \times N}}.$$
In this section, we develop an FP-based AO approach, which directly maximizes the spectral efficiency, to solve problem (P1). The main idea of FP is to liberate designed variables from the intractable ratio and logarithm terms. Subsequently, variables can be solved individually following the spirit of AO.
\subsection{Problem Reformulation} 
As discussed above, jointly obtaining the global optimum with respect to all variables in (P1) is intractable due to the highly coupled and non-convex structure of the problem. To make the problem more manageable, we reformulate it into a more tractable form. Specifically, although the transmit power constraint (\ref{1a}) is non-convex, it can be relaxed or equivalently transformed by exploiting the special structural properties of (P1).
\begin{lemma}[\emph{Power Constraint Relaxation}] \normalfont \label{lemma1}
Consider a new problem
\begin{subequations} \label{p2}
	\begin{align}
		\left( {\rm{P2}}\right)  \  \mathop {\max }\limits_{{{\bf{W}}_{{\rm{BB}}}},{{\bf{W}}_{{\rm{RF}}}},{\bf{X}}} \quad &\bar R = \sum\limits_{k = 1}^K {{\beta _k}{{\log }_2}\left( {1 + {{\bar \gamma }_k}} \right)} \label{1.5} \\
		\mathrm{s.t.} \ \  & (\ref{1aa}), (\ref{1b}), (\ref{1c}), 
	\end{align}
\end{subequations}
where 
\begin{equation}
	{{\bar \gamma }_k} = \frac{{{\rm{Tr}}( {{{\bf{H}}_k}{\bf{W}}{{\bf{E}}_k}{{\bf{W}}^H}} )}}{{{\rm{Tr}}( {{{\bf{H}}_k}{\bf{W}}\left( {{{\bf{I}}_K} - {{\bf{E}}_k}} \right){{\bf{W}}^H}}) + \frac{{\sigma _k^2}}{P}{{\left\| {{{\bf{W}}_{{\rm{RF}}}}{{\bf{W}}_{{\rm{BB}}}}} \right\|}_F^2}}}.
\end{equation}
Given the optimal solution for (P2) as ${\bf{W}}_{{\rm{BB}}}^\star$, ${\bf{W}}_{{\rm{RF}}}^\star$, and ${\bf{X}}^\star$, the optimal solution for (P1) is $\frac{{\sqrt P }}{{{{\left\| {{{\bf{W}}_{{\rm{RF}}}^\star}{{\bf{W}}_{{\rm{BB}}}^\star}} \right\|}_F}}}{\bf{W}}_{{\rm{BB}}}^\star$, ${\bf{W}}_{{\rm{RF}}}^\star$, and ${\bf{X}}^\star$.
\end{lemma} 
\begin{proof}
	Please refer to Appendix~\ref{appa}.
\end{proof}
\textbf{Lemma~\ref{lemma1}} allows the power constraint in (P1) to be equivalently absorbed into the objective function, leading to the reformulated problem (P2). Although this transformation slightly complicates the objective, it simplifies the constraint structure and preserves equivalence. Another challenge remaining in (P2) is the tangling of variables in the highly intractable sum–log–ratio objective. FP techniques, including the Lagrange dual transform~\cite[Theorem 3]{FP2} and the quadratic transform~\cite[Theorem 2]{FP1}, serve as an effective approach to decouple the variables and facilitate closed-form updates in each iteration. \\
\indent The Lagrangian dual transform introduces a non-negative real auxiliary variable ${\boldsymbol{\xi}} \in {\mathbb{R}^K}$ to release the signal and interference terms from the logarithm operation. Specifically, the objective (\ref{1.5}) of (P2) is equivalent to
\begin{align} \label{2}
	&\sum\limits_{k = 1}^K {{\beta _k}( {{{\log }_2}\left( {1 + {\xi _k}} \right) - {\xi _k}} )} \notag \\
	& \qquad \quad + \sum\limits_{k = 1}^K {\frac{{{\beta _k}\left( {1 + {\xi _k}} \right){\rm{Tr}}( {{{\bf{H}}_k}{\bf{W}}{{\bf{E}}_k}{{\bf{W}}^H}} )}}{{{\rm{Tr}}( {{{\bf{H}}_k}{\bf{W}}{{\bf{W}}^H}} ) + \frac{{\sigma _k^2}}{P}\left\| {{{\bf{W}}_{{\rm{RF}}}}{{\bf{W}}_{{\rm{BB}}}}} \right\|_F^2}}},
\end{align}
where the auxiliary variable is exactly the SINR of user $k$ and takes the form
\begin{equation} \label{xi}
	\xi _k = \frac{{{\rm{Tr}}( {{{\bf{H}}_k}{\bf{W}}{{\bf{E}}_k}{{\bf{W}}^H}} )}}{{{\rm{Tr}}( {{{\bf{H}}_k}{\bf{W}}\left( {{{\bf{I}}_K} - {{\bf{E}}_k}} \right){{\bf{W}}^H}} ) + \frac{{\sigma _k^2}}{P}\left\| {{{\bf{W}}_{{\rm{RF}}}}{{\bf{W}}_{{\rm{BB}}}}} \right\|_F^2}},
\end{equation}
By the quadratic transform, (\ref{2}) is further recast as
\begin{align} \label{3}
	&f_t = {\rm{const}}\left( {\boldsymbol{\xi }} \right) + \sum\limits_{k = 1}^K {2\sqrt {{\beta _k}\left( {1 + {\xi _k}} \right)} {\mathop{\rm Re}\nolimits} \left\{ {\mu _k^*{\bf{h}}_k^H{\bf{W}}{{\bf{e}}_k}} \right\}} \notag \\
	&- \sum\limits_{k = 1}^K {{{\left| {{\mu _k}} \right|}^2}\left( {{\rm{Tr}}( {{{\bf{H}}_k}{\bf{W}}{{\bf{W}}^H}}) + \frac{{\sigma _k^2}}{P}\left\| {{{\bf{W}}_{{\rm{RF}}}}{{\bf{W}}_{{\rm{BB}}}}} \right\|_F^2} \right)},
\end{align}
where ${\rm{const}}\left( {\boldsymbol{\xi }} \right) = \sum\nolimits_{k = 1}^K {{\beta _k}( {{{\log }_2}\left( {1 + {\xi _k}} \right) - {\xi _k}} )}$, and $\boldsymbol{\mu }$ is the second auxiliary variable introduced to linearize optimization variables from ratio expressions. The expression (\ref{3}) is equivalent to (\ref{2}) if and only if $\boldsymbol{\mu }$ takes the value
\begin{equation} \label{mu}
	{\mu _k} = \frac{{\sqrt {{\beta _k}\left( {1 + {\xi _k}} \right)} {\bf{h}}_k^H{\bf{W}}{{\bf{e}}_k}}}{{{\rm{Tr}}( {{{\bf{H}}_k}{\bf{W}}{{\bf{W}}^H}} ) + \frac{{\sigma _k^2}}{P}\left\| {{{\bf{W}}_{{\rm{RF}}}}{{\bf{W}}_{{\rm{BB}}}}} \right\|_F^2}}, \forall k \in {\mathcal{K}}
\end{equation}
Hitherto, (P2) is rearranged to a more tractable equivalent form:
\begin{subequations}
	\begin{align}
		\left( {\rm{P3}}\right)  \  \mathop {\max }\limits_{{\boldsymbol{\xi }}, {\boldsymbol{\mu }}, {{\bf{W}}_{{\rm{BB}}}},{{\bf{W}}_{{\rm{RF}}}},{\bf{X}}} \quad & f_t  \\
		\mathrm{s.t.} \ \  &{\boldsymbol{\xi }} \succeq \mathbf{0}_K, (\ref{1aa}), (\ref{1b}), (\ref{1c})
	\end{align}
\end{subequations}
In the rearranged problem (P3), it is still NP-hard to jointly find a global optimizer w.r.t. all variables. Fortunately, since the intractable logarithm and ratio operations are addressed, AO can be readily deployed to iteratively solve one variable while fixing the others. In particular, the overall optimization variables are divided into two blocks, namely the transmit beamforming matrices ${{\bf{W}}_{{\rm{BB}}}}$ and ${{\bf{W}}_{{\rm{RF}}}}$, and the pinching beamforming matrix ${\bf{X}}$. 

\subsection{Transmit Beamforming Optimization}
This section focuses on optimizing the transmit beamforming. Specifically, the digital beamformer ${{\bf{W}}_{{\rm{BB}}}}$ is obtained using a closed-form solution, while the analog beamformer ${{\bf{W}}_{{\rm{RF}}}}$ is optimized via the Riemannian conjugate gradient method, as detailed below.
\subsubsection{Digital Beamforming Optimization} The original optimization problem (P3) with respect to the digital beamforming matrix can be concisely expressed as 
\begin{align}
	\left( {\rm{P4}}\right) \ \mathop {\max }\limits_{{\bf{W}}_{{\rm{BB}}}} \  & 2{\mathop{\rm Re}\nolimits} \left\{ {{\rm{Tr}}( {{{\bf{A}}_{{\rm{BB}}}^H}{{\bf{W}}_{{\rm{BB}}}}} )} \right\} - {\rm{Tr}}( {{\bf{W}}_{{\rm{BB}}}^H{\bf{B}}_{{\rm{BB}}}{{\bf{W}}_{{\rm{BB}}}}})
\end{align}
where ${\bf{A}}_{{\rm{BB}}} = {\bf{W}}_{{\rm{RF}}}^H{\bf{W}}_{{\rm{PB}}}^H{\bf{\tilde H}}{\bf{C}}$, ${\bf{\tilde H}} = {\left[ {\mu_1{{\bf{h}}_1}, \cdots ,\mu_K{{\bf{h}}_K}} \right]}$, ${\bf{C}} = {\rm{diag}}\left( {\sqrt {{\beta _1}\left( {1 + {\xi _1}} \right)} , \cdots ,\sqrt {{\beta _K}\left( {1 + {\xi _K}} \right)} } \right)$, and ${\bf{B}}_{{\rm{BB}}} = {\bf{W}}_{{\rm{RF}}}^H\left( {{\bf{W}}_{{\rm{PB}}}^H{\bf{\tilde H}}{{{\bf{\tilde H}}}^H}{{\bf{W}}_{{\rm{PB}}}} + \frac{1}{P}\sum\nolimits_{k = 1}^K {{{\left| {{\mu _k}} \right|}^2}\sigma _k^2} {{\bf{I}}_M}} \right){{\bf{W}}_{{\rm{RF}}}}$. The sub-problem of ${{\bf{W}}_{{\rm{BB}}}}$ is essentially a unconstrained quadratic programming (QP) problem. Since ${\bf{B}}_{{\rm{BB}}}$ is surely positive semi-definite, (P4) is concave over ${{\bf{W}}_{{\rm{BB}}}}$. Consequently, the closed-form optimal solution for this sub-problem can be obtained via the first-order optimality condition as follows:
\begin{equation} \label{W_BB}
	{{\bf{W}}_{{\rm{BB}}}^\star} = {\bf{B}}_{{\rm{BB}}}^{ - 1}{{\bf{A}}_{{\rm{BB}}}}.
\end{equation}
\subsubsection{Analog Beamforming Optimization} The optimization sub-problem associated with the analog beamforming matrix can be reformulated in a compact form as follows
\begin{subequations} \label{W_RF}
	\begin{align}
		\left( {\rm{P5}}\right) \mathop {\max }\limits_{{{\bf{W}}_{{\rm{RF}}}}} \  & 2{\mathop{\rm Re}\nolimits} \left\{ {{\rm{Tr}}( {{{\bf{A}}_{{\rm{RF}}}^H}{{\bf{W}}_{{\rm{RF}}}}})} \right\} - {\rm{Tr}}( {{\bf{W}}_{{\rm{RF}}}^H{\bf{B}}_{{\rm{RF}}}{{\bf{W}}_{{\rm{RF}}}}}{\bf{Q}}_{{\rm{RF}}} )  \\
		\mathrm{s.t.} \ \  &(\ref{1aa}),
	\end{align}
\end{subequations}
where ${\bf{A}}_{{\rm{RF}}} = {\bf{W}}_{{\rm{PB}}}^H{\bf{\tilde HCW}}_{{\rm{BB}}}^H$, ${\bf{Q}}_{{\rm{RF}}} = {\bf{W}}_{{\rm{BB}}}{{\bf{W}}_{{\rm{BB}}}^H}$, and ${\bf{B}}_{{\rm{RF}}} = {\bf{W}}_{{\rm{PB}}}^H{\bf{\tilde H}}{{{\bf{\tilde H}}}^H}{{\bf{W}}_{{\rm{PB}}}} + \frac{1}{P}\sum\nolimits_{k = 1}^K {{{\left| {{\mu _k}} \right|}^2}\sigma _k^2} {{\bf{I}}_M}$. It is observed that (P5) exhibits a similar structure with (P4) in terms of the objective function and hence is also a QP over ${\bf{W}}_{{\rm{RF}}}$. Since ${\bf{B}}_{{\rm{RF}}}$ and ${\bf{Q}}_{{\rm{RF}}}$ are positive semi-definite, the objective function is concave on ${\bf{W}}_{{\rm{RF}}}$. However, the unit modulus constraint on the elements of ${\bf{W}}_{{\rm{RF}}}$ confines a highly non-convex set, which renders the solving of (P5) intractable. \\
\indent The unit modulus constraint essentially regulates the analog beamformer on a complex torus, which is the Cartesian product of $M \times N_{RF}$ complex unit circles. To tackle this constraint, common approaches include the projected gradient descent (PGD)~\cite{Boyd, PGD},  MO~\cite{ManOpt, ManOpt2}, and extreme point pursuit (EPP)~\cite{EPP1, EPP2}. However, since the complex torus is a non-convex set, the convergence of PGD is not guaranteed. The EPP algorithm, which is also centered around the PGD, adopts the spirit of penalty and homotopy to induce a solution for constant modulus problems. Nonetheless, the convergence is still not theoretically guaranteed and verified by merely simulation results. The setting of penalty sequence also requires elaborate tuning. Conversely, a stationary solution can be obtained by the MO if the objective function is smooth. Consequently, the MO is the most suitable approach to solve (P5). Specifically, since the objective (\ref{W_RF}) is a concave quadratic function of ${\bf{W}}_{{\rm{RF}}}$ with satisfactory properties, the Riemannian conjugate gradient method is suitable to obtain a stationary solution for (P5). \\
\indent Denote the feasible set of analog beamformers as $\cal W$. $\cal W$ is a Riemannian submanifold of $\mathbb{C} \in ^{M \times N_{{\rm{RF}}}}$ equipped with the canonical metric inherited from the Euclidean space and (\ref{W_RF}) as $f_{{\rm{RF}}}$. The Euclidean gradient of (\ref{W_RF}) is given by
\begin{equation} \label{grad1}
	\nabla {f_{{\rm{RF}}}}\left( {{{\bf{W}}_{{\rm{RF}}}}} \right) = {{\bf{A}}_{{\rm{RF}}}} - {{\bf{B}}_{{\rm{RF}}}}{{\bf{W}}_{{\rm{RF}}}}{{\bf{Q}}_{{\rm{RF}}}}.
\end{equation}
The Riemannian gradient at a given point is then defined as the projection of the Euclidean gradient onto the tangent space $\mathcal{T}_{{\bf{W}}_{{\rm{RF}}}}\mathcal{W}$, i.e.,
\begin{align} \label{grad2}
	&{\rm{grad}}{f_{{\rm{RF}}}}\left( {{{\bf{W}}_{{\rm{RF}}}}} \right) = {\rm{Pro}}{{\rm{j}}_{\mathcal{T}_{{\bf{W}}_{{\rm{RF}}}}\mathcal{W}}}\left( {\nabla {f_{{\rm{RF}}}}\left( {{{\bf{W}}_{{\rm{RF}}}}} \right)} \right)  \notag \\
	&=\nabla {f_{{\rm{RF}}}}\left( {{{\bf{W}}_{{\rm{RF}}}}} \right) - {\mathop{\rm Re}\nolimits} \left\{ {\nabla {f_{{\rm{RF}}}}\left( {{{\bf{W}}_{{\rm{RF}}}}} \right) \odot {\bf{W}}_{{\rm{RF}}}^*} \right\} \odot {{\bf{W}}_{{\rm{RF}}}},
\end{align}
where $\mathcal{T}_{{\bf{W}}_{{\rm{RF}}}}\mathcal{W} = \left\{ {{\bf{Z}} \in \mathbb{C} ^{M \times N_{{\rm{RF}}}} \left| {{\mathop{\rm Re}\nolimits} \left\{ {{\bf{Z}} \odot {\bf{W}}_{{\rm{RF}}}^*} \right\} = {\bf{0}}} \right.} \right\}$. \\
\indent Given the Riemannian gradient, the search direction ${{\bf{D}}_t}$ at the $t$-th iteration is determined by combining the negative Riemannian gradient with the previous direction ${{\bf{D}}_{t-1}}$ through a conjugacy parameteras $\rho_t$. This process follows the spirit of classical conjugate gradient method. Formally,
\begin{equation} \label{conjugacy}
	{{\bf{D}}_{t+1}} =  - {\rm{grad}}{f_{{\rm{RF}}}}\left( {{\bf{W}}_{{\rm{RF}}}^{\left( t+1 \right)}} \right) + {\rho _{t+1}}{\mathcal{T}_{{\bf{W}}_{{\rm{RF}}}^{\left( t+1 \right)}}}\left( {{{\bf{D}}_{t}}} \right),
\end{equation}
where ${\mathcal{T}_{{\bf{W}}_{{\rm{RF}}}^{\left( t \right)}}}$ denotes the \textit{transport} operator mapping a tangent vector from $\mathcal{T}_{{\bf{W}}_{{\rm{RF}}}^{\left( t-1 \right)}}\mathcal{W}$ to $\mathcal{T}_{{\bf{W}}_{{\rm{RF}}}^{\left( t \right)}}\mathcal{W}$. This \textit{transport} can also be realized via the projection, i.e., 
\begin{equation}
	{\mathcal{T}_{{\bf{W}}_{{\rm{RF}}}^{\left( t \right)}}}\left( {{{\bf{D}}_{t - 1}}} \right) = {{\bf{D}}_{t - 1}} - {\mathop{\rm Re}\nolimits} \left\{ {{{\bf{D}}_{t - 1}} \odot {{\left( {{\bf{W}}_{{\rm{RF}}}^{\left( t \right)}} \right)}^*}} \right\} \odot {\bf{W}}_{{\rm{RF}}}^{\left( t \right)}.
\end{equation}
The Polak–Ribière rule~\cite{nonlinear} is typically adopted for $\rho_t$. \\
\begin{algorithm}[tbp]
	\setlength{\textfloatsep}{0.cm}
	\setlength{\floatsep}{0.cm}
	\small
	\caption{Riemannian conjugate gradient algorithm for analog beamforming optimization}
	\renewcommand{\algorithmicrequire}{\textbf{Input}}
	\renewcommand{\algorithmicensure}{\textbf{Output}}
	\label{alg1}
	\begin{algorithmic}[1]
		\REQUIRE ${\bf{W}}_{{\rm{RF}}}^{\left( 0 \right)}$, ${{\bf{A}}_{{\rm{RF}}}}$, ${{\bf{B}}_{{\rm{RF}}}}$, ${{\bf{Q}}_{{\rm{RF}}}}$, and $t=0$.
		\ENSURE A stationary solution ${\bf{W}}_{{\rm{RF}}}^{\star}$.
		\STATE ${{\bf{D}}_0} =  - {\rm{grad}}{f_{{\rm{RF}}}}\left( {{\bf{W}}_{{\rm{RF}}}^{\left( 0 \right)}} \right)$;
		\REPEAT
		\STATE Choose Armijo backtracking line search step size $\nu_{t}$;
		\STATE Update the solution via retraction by (\ref{retraction});
		\STATE Compute the new Riemannian gradient by (\ref{grad1}) and (\ref{grad2});
		\STATE Calculate the conjugate direction by (\ref{conjugacy});
		\STATE $k \leftarrow k+1$;
		\UNTIL a stopping criterion meets.
	\end{algorithmic}
\end{algorithm}	
\indent After moving along the search direction, the updated point must be mapped back to the manifold to satisfy the unit-modulus constraint. This is achieved via the \textit{retraction} operation. Specially, this process is represented as
\begin{equation} \label{retraction}
	{\bf{W}}_{{\rm{RF}}}^{\left( t+1 \right)} = {\mathcal{R}_{{\bf{W}}_{{\rm{RF}}}^{\left( {t} \right)}}}\left( {{\nu  _t}{{\bf{D}}_t}} \right) = \exp \{ {j\arg \left( {{\bf{W}}_{{\rm{RF}}}^{\left( {t} \right)} + {\nu _t}{{\bf{D}}_t}} \right)} \},
\end{equation}
where $\mathcal{R}$ is the \textit{retraction} operator and $\arg$ is an operation to extract the phase element-wisely. This procedure is summarized in \textbf{Algorithm~\ref{alg1}} and theoretically guaranteed to converge to a stationary solution. It is worthy to mention that since a closed-form stationary solution can be calculated for (\ref{W_RF}), i.e., ${\bf{W}}_{{\rm{RF}}}^{{\rm{sta}}} = {\bf{B}}_{{\rm{RF}}}^{ - 1}{{\bf{A}}_{{\rm{RF}}}}{\bf{Q}}_{{\rm{RF}}}^{ - 1}$, ${\bf{W}}_{{\rm{RF}}}$ can initialized by projecting this solution to $\cal W$. 

\subsection{Pinching Beamforming Optimization}
After determining the updating rule of two transmit beamformers, ${\bf{W}}_{{\rm{BB}}}$ and ${\bf{W}}_{{\rm{RF}}}$, in an alternating manner, the third beamforming DoF, pinching beamforming, together with the effective channel, is controlled by the position of PAs, which is critical in the proposed FC-THB PASS. 
\subsubsection{Multi-Modal Optimization Challenges} 
Treating the positions individually, the joint pinching beamforming and channel optimization sub-problem is represented as 
\begin{subequations}\label{W_PB}
	\begin{align} 
		\left( {\rm{P6}}\right) \quad \mathop {\max }\limits_{{{\bf{X}}}} \quad  &f_X =  2{\mathop{\rm Re}\nolimits} \left\{ {{\rm{Tr}}( {{\bf{\bar G}}\left( {\bf{X}} \right){{\bf{A}}_X}{\bf{C}}} )} \right\} \notag \\ 
		&- {\rm{Tr}}( {{\bf{\bar G}}\left( {\bf{X}} \right){{\bf{B}}_X}{{{\bf{\bar G}}}^H}\left( {\bf{X}} \right)} ) \\
		\mathrm{s.t.} \ \  &(\ref{1b}), (\ref{1c}),
	\end{align}
\end{subequations}
where ${\bf{A}}_X = {\bf{W}}_{{\rm{RF}}}{\bf{W}}_{{\rm{BB}}}$, ${\bf{B}}_X = {\bf{A}}_{{\rm{X}}}{\bf{A}}_{{\rm{X}}}^H$, and ${\bf{\bar G}}\left( {\bf{X}} \right) = {{{\bf{\tilde H}}}^H}\left( {\bf{X}} \right){{\bf{W}}_{{\rm{PB}}}}\left( {\bf{X}} \right) = {\left[ {\mu _1^*{{{\bf{\bar g}}}_1}\left( {\bf{X}} \right), \cdots ,\mu _K^*{{{\bf{\bar g}}}_K}\left( {\bf{X}} \right)} \right]^T}$. By the signal propagation model in (\ref{W_PB_L}) and (\ref{H_L}), $${{{\bf{\bar g}}}_k}\left( {\bf{X}} \right) = {\bf{h}}_k^H\left( {\bf{X}} \right){{\bf{W}}_{{\rm{PB}}}}\left( {\bf{X}} \right) = \left[ {{{\bar g}_{k,1}}\left( {{{\bf{x}}_1}} \right), \cdots ,{{\bar g}_{k,M}}\left( {{{\bf{x}}_M}} \right)} \right],$$
where
\begin{align}
	&{{\bar g}_{k,m}}\left( {{{\bf{x}}_m}} \right) = \notag \\ &\sum\limits_{n = 1}^N {\frac{{\eta\exp \left\{ { - j\frac{{2\pi }}{\lambda }\left( {{D_{k,m,n}}\left( {x_{m,n}} \right)  + {n_{{\rm{eff}}}}{x_{m,n}}} \right)} \right\}}}{{{\sqrt N}{D_{k,m,n}}\left( {x_{m,n}} \right) }}}.
\end{align}
\indent It can be observed from $f_X$ that the impact of PA position adjustment affects both the phase shift and amplitude attenuation in the signal propagation process of PASS. The resulting objective function is essentially \textit{multi-modal}~\cite{multi-modal}, containing numerous stationary points corresponding to local maxima or minima. Each user receives $M \times N$ attenuated and phase-shifted replicas of the transmitted signal, which collectively create a highly intricate landscape for the objective function, making this sub-problem extremely challenging to solve optimally. An example is provided in Appendix~\ref{challenge} to illustrate this difficulty. \\
\indent As illustrated by Fig.~\ref{fig_challenge}, the target objective function presents a rugged surface characterized by numerous peaks and valleys. This type of \textit{multi-modal} problem widely exists in the physical world. prevalent in physical systems. Conventional gradient descent methods often become trapped at suboptimal stationary points. Similarly, majorization and approximation-based techniques encounter comparable challenges. To date, most PASS frameworks have employed a PA-wise line grid search, which yields near coordinate-wise optimal solutions but is computationally intensive, necessitating a trade-off between optimality and computational complexity. In this paper, an evolution-based algorithm is proposed to arrive a near globally optimal solution for ${{\bf{X}}}$.

\subsubsection{Evolution-based PA Position Optimization}
It has been widely investigated in the applied math community that the evolutionary algorithms and their variants are particularly useful in this type of problem due to their population-based approach and inherent flexibility. Consequently, the SHADE algorithm~\cite{SHADE} is adopted to tackle the optimization of (P6). SHADE evolves a population of candidate solutions through iterative processes of mutation, crossover, and selection, leveraging these operations to converge toward an optimal solution. \\
\begin{algorithm}[tbp]
	\setlength{\textfloatsep}{0.cm}
	\setlength{\floatsep}{0.cm}
	\small
	\caption{SHADE algorithm for PA position optimization}
	\renewcommand{\algorithmicrequire}{\textbf{Input}}
	\renewcommand{\algorithmicensure}{\textbf{Output}}
	\label{alg2}
	\begin{algorithmic}[1]
		\REQUIRE $P_{\rm{DE}}$, $p$, max generation $G_{\rm{DE}}$, memory size $H_{\rm{DE}}$, and essential parameters for evaluating the objective function.
		\ENSURE A near globally optimal solution ${\bf{X}}^{\star}$.
		\STATE Initialize population $\left\{ {{\bf{X}}_i^{\left( 0 \right)}} \right\}_{i = 1}^{P_{\rm{DE}}}$ randomly in feasible space, memories $\left\{ {{\mu_{F,h}},\mu_{{\rm{CR}},h}} \right\}_{h = 1}^H$ all as 0.5, and archive ${\cal A} = \emptyset$;
		\FOR {generation $g_{\rm{DE}}=1:G_{\rm{DE}}$}
		\FOR {each candiate ${\bf{X}}_i^{\left( g_{\rm{DE}}-1 \right)}$}
		\STATE Choose memory index $k$ randomly and sample ${F_i} \sim \text{Cauchy} \left( {\mu_{{F},k},0.1} \right), {{\rm{CR}}_i} \sim \mathcal{N} \left( {\mu_{{\rm{CR}},k},0.1} \right)$ at $\left( {0,1} \right]$;
		\STATE Select ${{\bf{X}}_{{\rm{pbest}}}}$, ${{\bf{X}}_{{r_1}}}$, and ${{\bf{X}}_{{r_2}}}$ from the top $p \% $, population, and union of population and ${\cal A}$;
		\STATE Mutate by (\ref{mutation});
		\STATE Crossover by (\ref{crossover});
		\IF {$f_X({\bf{U}}_i) > f_X({\bf{X}}_i^{\left( g_{\rm{DE}}-1 \right)})$} 
		\STATE ${\bf{X}}_i^{\left( g_{\rm{DE}} \right)} = {{\bf{U}}_i}$, $A \leftarrow A \cup \left\{ {{{\bf{U}}_i}} \right\}$, and record the success index $i$, parameter ${F_i}, {{\rm{CR}}_i}$, and improvement $\Delta {f_i} = {f_X}\left( {{{\bf{U}}_i}} \right) - {f_X}\left( {\bf{X}}_i^{\left( g_{\rm{DE}}-1 \right)} \right)$;
		\ELSE 
		\STATE ${\bf{X}}_i^{\left( g_{\rm{DE}} \right)}={\bf{X}}_i^{\left( g_{\rm{DE}}-1 \right)}$
		\ENDIF
		\ENDFOR
		\STATE Collect success indices as $\cal S$. If $\cal S$ is not empty, update the memory slot indexed by $h = \left( {g_{\rm{DE}} \bmod H_{\rm{DE}}} \right)+1$ as $$\mu_{F,k} \leftarrow \frac{{\sum\nolimits_{s \in \mathcal{S}} {{w_s}F_s^2} }}{{\sum\nolimits_{s \in \mathcal{S}} {{w_s}{F_s}} }},\mu_{{\rm{CR}},k} \leftarrow \sum\nolimits_{s \in \mathcal{S}} {{w_s}{{\rm{CR}}_s}},$$ where ${w_s}$ is the normalized weight ${w_s} = \frac{{\Delta {f_s}}}{{\sum\nolimits_{t \in \mathcal{S}} {\Delta {f_t}} }},s \in \mathcal{S}$;
		\STATE If $\left| \mathcal{A} \right| > P_{\rm{DE}}$, randomly remove elements until $\left| \mathcal{A} \right| = P_{\rm{DE}}$;
		\STATE $g_{\rm{DE}} \leftarrow g_{\rm{DE}}+1$;
		\ENDFOR
	\end{algorithmic}
\end{algorithm}	
\indent Specifically, the first step for SHADE is to initialize a population of $P_{\rm{DE}}$ candidate solutions $\left\lbrace {{\bf{X}}_1}, {{\bf{X}}_2}, \cdots, {{\bf{X}}_{P_{\rm{DE}}}}\right\rbrace $, each sampled from the feasible space of ${{\bf{X}}}$. For each candidate, the objective value is evaluated, serving as the fitness score. Second, the \textit{mutation} is performed. For each candidate solution, SHADE generates a mutant using the current-to-$p$best mutation strategy
\begin{equation} \label{mutation}
	{{\bf{V}}_i} = {{\bf{X}}_i} + F\left( {{{\bf{X}}_{{\rm{pbest}}}} - {{\bf{X}}_i}} \right) + F\left( {{{\bf{X}}_{{r_1}}} - {{\bf{X}}_{{r_2}}}} \right).
\end{equation}
where ${{\bf{X}}_{{\rm{pbest}}}}$ is randomly chosen from the top $p \% $ of the current population, i.e., elite solutions, and ${{\bf{X}}_{{r_1}}}$, ${{\bf{X}}_{{r_2}}}$ are distinct randomly selected individuals. The mutation factor $F$ controls the step size and is adaptively tuned. A trial solution is then formed by combining the mutant with the original candidate via \textit{crossover} to increase diversity
\begin{equation} \label{crossover}
	{u_{i,nm}} = \left\{ \begin{array}{l}
		{v_{i,nm}}, \quad \text{if} \quad \text{cr} < \text{CR}\\
		{l_{i,nm}}, \quad \text{otherwise}
	\end{array} \right. , \forall nm
\end{equation}
where $\text{CR}$ is the crossover rate, $\text{cr}$ is a number drawn from the uniform distribution between 0 and 1, and the subscript $nm$ denotes the $(n,m)$-th element of the position matrix \footnote{Note that during the mutation and crossover stages, both the mutant and descendant solutions must remain within the feasible region to ensure compliance with the problem's constraints.}. Finally, the trial solution competes with its parent, i.e., the solution yielding a large objective function value survives. This procedure is conducted among the population and only solutions that improve or maintain the fitness survive to the next generation. \\
\indent The essence of SHADE lies in its success-history based adaptation of control parameters to enhance the performance and robustness of differential evolution algorithm. In particular, it maintains an external memory archive of historical successful values of $F$ and $\text{CR}$. At each generation, new parameter values are sampled from probability distributions whose means are updated using this archive. This allows SHADE to gradually learn parameter settings that are well-suited to the problem structure, improving convergence and stability. The detailed SHADE for acquiring a near globally optimal position for PAs' positions is summarized in \textbf{Algorithm~\ref{alg2}}.
%Commonly used evolutionary or nature-inspired algorithms include genetic algorithms, particle swarm optimization, and differential evolution (DE), to name a few. Among these, DE is particularly attractive due to its simple structure, few control parameters, and strong capability in handling continuous and non-convex optimization problems. For this reason, we employ DE in this paper to address the optimization of (P6). However, classical DE may suffer from the risk of premature convergence in multi-modal landscapes. To overcome these issues, we further adopt the success-history based adaptive DE (SHADE) algorithm, which enhances DE by introducing adaptive parameter control and the current-to-$p$best mutation strategy. These modifications allow SHADE to dynamically balance exploration and exploitation, maintain population diversity, and achieve more robust convergence behavior, making it well suited for optimizing the highly non-convex objective function 
%In the context of our problem, each candidate L represents an N×M matrix configuration that defines , and thereby determines the objective . Through successive generations, SHADE adaptively refines  by guiding the population toward elite candidates while continuously adjusting its mutation and crossover behavior based on past successes. This process allows SHADE to effectively navigate the highly nonlinear and oscillatory landscape of , avoiding premature stagnation while steadily approaching the global optimum or high-quality near-optimal solutions

\subsection{Overall Algorithm, Convergence, and Complexity}
Building upon the solutions developed for each variable, and the result the \textbf{Lemma~\ref{lemma1}}, the overall AO algorithm for solving (P1) is summarized in \textbf{Algorithm~\ref{alg3}}. 
\begin{algorithm}[tbp]
	\setlength{\textfloatsep}{0.cm}
	\setlength{\floatsep}{0.cm}
	\small
	\caption{FP-based AO for FC-THB PASS}
	\renewcommand{\algorithmicrequire}{\textbf{Input}}
	\renewcommand{\algorithmicensure}{\textbf{Output}}
	\label{alg3}
	\begin{algorithmic}[1]
		\REQUIRE System parameters.
		\ENSURE Near coordinate-wisely optimal ${\bf{W}}_{{\rm{BB}}}^{\star}$, ${\bf{W}}_{{\rm{RF}}}^{\star}$, ${\bf{X}}^{\star}$.
		\STATE Initialize ${\bf{W}}_{{\rm{BB}}}$, ${\bf{W}}_{{\rm{RF}}}$, ${\bf{X}}$;
		\REPEAT
		\STATE Update $\boldsymbol{\xi}$ by (\ref{xi});
		\STATE Update $\boldsymbol{\mu}$ by (\ref{mu});
		\STATE Update ${\bf{W}}_{{\rm{BB}}}$ by (\ref{W_BB});
		\STATE Update ${\bf{W}}_{{\rm{RF}}}$ following \textbf{Algorithm~\ref{alg1}};
		\STATE Update ${\bf{X}}$ following \textbf{Algorithm~\ref{alg2}};
		\UNTIL a stopping criterion meets.
		\STATE Scale the digital beamformer as ${\bf{W}}_{\rm{BB}} \leftarrow \frac{{\sqrt P }}{{{{\left\| {{{\bf{W}}_{{\rm{RF}}}}{{\bf{W}}_{{\rm{BB}}}}} \right\|}_F}}}{\bf{W}}_{\rm{BB}}$.
	\end{algorithmic}
\end{algorithm}	
\begin{table}[tbp]
	\small
	\centering
	\caption{Comparison of different position optimization methods}
	\label{tab1}
	\begin{tabular}{ccc}
		\toprule
		\textbf{Method}      & \textbf{Optimality}                     & \textbf{Complexity} \\
		\midrule
		GD/SCA      & optimality-unknown stationary  &           ${i_X}{C_X}$ \\
		\midrule
		Grid Search & near coordinate-wisely optimal &           ${{MN}{N_{{\rm{grid}}}}{C_X}}$ \\
		\midrule
		SHADE       & near globally optimal          &           ${{i_X}{P_{{\rm{DE}}}}{C_X}}$ \\
		\bottomrule
	\end{tabular}
\end{table}
In each outer iteration of \textbf{Algorithm~\ref{alg3}}, ${\bf{W}}_{{\rm{BB}}}$ is updated by a closed-form solution. A stationary ${\bf{W}}_{{\rm{RF}}}$ can be acquired via \textbf{Algorithm~\ref{alg1}}~\cite{ManOpt} and the SHADE arrives a near optimal ${\bf{X}}$. In other words, the value of objective function is ensured to be non-decreasing during AO iterations. Furthermore, the WSR is inherently upper-bounded due to the transmit power constraint. Hence, the convergence of the proposed AO algorithm is guaranteed. \\
\indent In each iteration of \textbf{Algorithm~\ref{alg3}}, the computational complexity for optimizing ${\bf{W}}_{{\rm{BB}}}$, $\boldsymbol{\xi}$, and $\boldsymbol{\mu}$ are given by $\mathcal{O} \left( {{M^3}{N^2} + {M^2}{N^2}K + N_{{\rm{RF}}}^3} \right)$, $\mathcal{O} \left( {{M^2}{N^2}{K^2}} \right)$, and $\mathcal{O} \left( {{M}{N}{K^2}} \right)$, respectively. The complexity for optimizing ${\bf{W}}_{{\rm{RF}}}$ lies in the calculation of gradient and objective function in each inner iteration in \textbf{Algorithm~\ref{alg1}} and hence is $\mathcal{O} \left( {I_{\mathrm{iter}, 1}{M^2}{N^2}} \right)$, where $I_{\mathrm{iter}, 1}$ is the number of iterations for Algorithm~\ref{alg1}. For the optimization of ${\bf{X}}$, the most expensive cost is evaluating the achieved objective value for a candidate solution. The complexity for such a evaluation is $C_X = \mathcal{O} \left( {KMN + KM\left( {K + M} \right)} \right)$. The proposed SHADE requires ${I_{\mathrm{iter}, 2}{P_{{\rm{DE}}}}{C_X}}$ computational resources, where $I_{\mathrm{iter}, 2}$ is the iteration number of \textbf{Algorithm~\ref{alg2}}. In the case, $M$, $N$, and $K$ are massive, the overall complexity for \textbf{Algorithm~\ref{alg3}} can be approximated as $\mathcal{O} \left( {{M^2}{N^2}{K^2}} \right)$. Otherwise, the complexity is dominated by the Riemannian conjugate gradient method and SHADE. \\
\indent Furthermore, a comparison of commonly used position optimization methods is presented in Table~\ref{tab1}. The gradient computation cost is equal to the function cost, i.e., $C_X$, and ${N_{{\rm{grid}}}}$ is the number of discretized grids. Although GD/SCA incur relatively low computational costs, the optimality is not guaranteed. On the other hand, the grid search method is both computationally expensive and less effective than SHADE, since ${N_{\mathrm{grid}}}$ must be sufficiently large to ensure that the coordinate-wise optimum is closely approximated.

\section{Tri-Hybrid Beamforming via Zero Forcing} \label{sec4}
Although the algorithm presented in the previous section can directly maximize the achievable WSR of the FC-THB PASS, it may involve prohibitively high computational complexity with a large number of waveguides and PAs, hindering practical implementation. This motivates the development of a beamforming and position optimization algorithm with lower computational complexity with tolerable performance loss. 
\subsection{ZF Beamforming}
To reduce the complexity of designing the transmit beamformer, zero-forcing (ZF) beamforming is adopted in this section. ZF is a widely used linear precoding technique in multi-user MIMO systems, known for its simplicity and robust performance. By employing a fixed structure, i.e., the pseudo-inverse of the channel matrix, ZF ensures that each user receives only their intended signal, effectively eliminating inter-user interference with low computational complexity~\cite{ZF}. In particular, treating the pinching beamforming ${\bf{W}}_{{\rm{PB}}}({\bf{X}})$ and PA-user propagation ${\bf{H}}({\bf{X}})$ together as the effective channel, i.e., ${\bf{G}}_{\text{ZF}}({\bf{X}})={\bf{W}}_{{\rm{PB}}}^H({\bf{X}}){\bf{H}}({\bf{X}})$, the ZF transmit beamformer is given by~\cite{ZF}
\begin{equation}
	\mathbf{W}_{\text{ZF},0}(\mathbf{X}) = \mathbf{G}_{\text{ZF}}(\mathbf{X}) \left( \mathbf{G}_{\text{ZF}}^H(\mathbf{X}) \mathbf{G}_{\text{ZF}}(\mathbf{X}) \right)^{-1}, \notag
\end{equation}
which is also a function of the PAs' positions $\mathbf{X}$. In other words, for any given PA position configuration $\mathbf{X}$, there exists a corresponding ZF beamformer. Furthermore, since the transmit power is restricted, the deployed ZF beamformer should be normalized to $\mathbf{W}_{\text{ZF}} = {\textstyle{{\sqrt P } \over {\left\| {{{\bf{W}}_{{\text{ZF}},0}}} \right\|_F}}}{{\bf{W}}_{{\text{ZF}},0}}$. With the normalized ZF transmit beamforming, the inter-user interference is entirely eliminated and the down-link WSR is hence given by
\begin{equation}
	{R_{\text{ZF}}}\left( {\bf{X}} \right) = \sum\limits_{k = 1}^K {\beta_k \log \left( {1 + \frac{{\left| {{\bf{g}}_{{\text{ZF}},k}^H\left( {\bf{X}} \right){{\bf{w}}_{{\text{ZF}},k}\left( {\bf{X}} \right)}} \right|^2}}{{\sigma _k^2}}} \right)} \notag,
\end{equation}
where ${\bf{g}}_{{\text{ZF}},k}$ and ${\bf{w}}_{{\text{ZF}},k}$ are the $k$-th column of ${\bf{G}}_{\text{ZF}}$ and ${\bf{W}}_{\text{ZF}}$, respectively. By exploring the relationship between $\mathbf{W}_{\text{ZF}}$ and $\mathbf{G}_{\text{ZF}}$, the expression of ${R_{\text{ZF}}}$ can be recast as 
\begin{equation}
	{R_{\text{ZF}}}\left( {\bf{X}} \right) = \sum\limits_{k = 1}^K {\beta_k \log \left( {1 + \frac{{\underline h}_k{p_k}}{{\sigma _k^2}}} \right)}, 
\end{equation}
where $p_k$ is the power allocated to the $k$-th user and ${\underline h}_k=1/\left[ {{\left( {\bf{G}}_\text{ZF}^H {\bf{G}}_{\text{ZF}} \right)}^{ - 1}} \right]_{k,k}$ indicates the quality of the $k$-th user's channel. The problem to maximize the WSR is then formulated as
\begin{subequations}
	\begin{align}
		\left( {\rm{P7}}\right) \quad \mathop {\max }\limits_{{\left\{ {{p_k}} \right\}_{k \in \cal K}}} \quad  &R_{\text{ZF}}  \\ 
		\mathrm{s.t.} \ \  &\sum\nolimits_{k \in K} {{p_k}}  = P,{p_k} \ge 0,\forall k.
	\end{align}
\end{subequations}
This problem is classical in traditional wireless communication and the solution for $p_k$ is exactly determined by the water-filling algorithm~\cite{Boyd}, i.e., $p_k^ \star  = {\left( {{\beta _k}\nu  - {\textstyle{{\sigma _k^2} \over \underline{h}_k}}} \right)^ + }$. If the set of activated users is $\mathcal{P}$, $\nu  = {\textstyle{{P + \sum\nolimits_{k \in \mathcal{P}} {\sigma _k^2/{\underline{h}_k}} } \over {\sum\nolimits_{k \in \mathcal{P}} {{\beta _k}} }}}$. Consequently, the maximum achievable WSR is obtained as
\begin{equation}
	{{\bar R}_{\text{ZF}}}\left( {\bf{X}} \right) = \sum\limits_{k \in \mathcal{P}} {{\log \left( {\frac{{{\beta _k}{\underline{h}_k}}}{{\sigma _k^2}}\nu } \right)}}.
\end{equation}

\subsection{Transmit Beamforming Optimization}
Given the ZF beamformer, we first consider its implementation within the transmit beamforming structure of the proposed FC-THB PASS, which comprises both digital and analog beamformers. Therefore, the ZF beamformer must be realized through the cascade of these two components, leading to a matrix decomposition problem formulated as
\begin{subequations}
	\begin{align}
		\left( {\rm{P9}}\right) \quad \mathop {\min}\limits_{{\bf{W}}_{{\rm{RF}}},{{\bf{W}}_{{\rm{BB}}}}} &\left\| {{{\bf{W}}_{{\rm{ZF}}}} - {{\bf{W}}_{{\rm{RF}}}}{{\bf{W}}_{{\rm{BB}}}}} \right\|_F^2 \label{decomposition} \\
		\mathrm{s.t.} \ \  &(\ref{1a}), {\bf{W}}_{{\rm{RF}}} \in \mathcal{W}. 
	\end{align}
\end{subequations}
In the problem formulated above, the constraint (\ref{1a}) can be temporarily ignored during the optimization, and ${\bf W}_{\rm BB}$ can be normalized afterward. This is supported by the fact that as long as the Euclidean distance between the optimal ZF beamformer and the transmit beamformers are made sufficiently small when ignoring the power constraint, the normalization step will also achieve a small distance to the ZF beamformer~\cite[Lemma 1]{ManOpt2}. \\
\indent Since ${\bf{W}}_{{\rm{RF}}}$ and ${{\bf{W}}_{{\rm{BB}}}}$ remain coupled in the objective function, an AO strategy is still required. By fixing ${\bf{W}}_{{\rm{RF}}}$, the digital beamformer can be obtained by solving $$\mathop {\arg \min }\limits_{{{\bf{W}}_{{\rm{BB}}}}} \left\| {{{\bf{W}}_{{\rm{ZF}}}} - {{\bf{W}}_{{\rm{RF}}}}{{\bf{W}}_{{\rm{BB}}}}} \right\|_F^2,$$ which is a standard least-squares (LS) problem. The corresponding closed-form solution is 
\begin{equation} \label{hbwbb}
	{{\bf{W}}_{{\rm{BB}}}^\star} = {\bf{W}}_{{\rm{RF}}}^\dag {{\bf{W}}_{{\rm{ZF}}}},
\end{equation} 
where ${\bf{A}}^\dag$ denotes the Moore-Penrose inverse of ${\bf{A}}$. \\
\indent Given ${{\bf{W}}_{{\rm{BB}}}}$, the sub-problem w.r.t. ${{\bf{W}}_{{\rm{RF}}}}$ is formulated as $$\mathop {\arg \min }\limits_{{{\bf{W}}_{{\rm{RF}}}} \in \mathcal{W}} \left\| {{{\bf{W}}_{{\rm{ZF}}}} - {{\bf{W}}_{{\rm{RF}}}}{{\bf{W}}_{{\rm{BB}}}}} \right\|_F^2,$$ which is also an LS problem but subject to a unit-modulus constraint. The Riemannian conjugate gradient-based MO method discussed in Section~\ref{sec3} can efficiently handle this constraint and obtain a stationary solution over the complex torus. For this problem, the Euclidean gradient is given by $$\nabla f_{\text{RF,LS}}= {{\bf{W}}_{{\rm{RF}}}}{{\bf{W}}_{{\rm{BB}}}}{\bf{W}}_{{\rm{BB}}}^H - {{\bf{W}}_{{\rm{ZF}}}}{\bf{W}}_{{\rm{BB}}}^H.$$ Replacing $\nabla f_{\text{RF}}$ by $\nabla f_{\text{RF,LS}}$, the detailed optimization procedure follows \textbf{Algorithm~\ref{alg1}}.

\subsection{Pinching Beamforming Optimization}
Note that ${\underline h}_k$, $\nu$, ${{\bar R}_{\text{ZF}}}$, and ${{\bf{W}}_{{\rm{ZF}}}}$ are functions of the PAs' positions. Additionally, the activated user set $\cal P$ is also influenced by $\mathbf{X}$. Thus, the optimal PA positions $\mathbf{X}$ must be determined to maximize the WSR of the proposed FC-THB PASS. The problem for maximizing ${{\bar R}_{\text{ZF}}}$ is formulated as
\begin{subequations}
	\begin{align}
		\left( {\rm{P8}}\right) \quad \mathop {\max }\limits_{{{\bf{X}}}} \quad  &{\bar R}_{\text{ZF}} \left( {\bf{X}} \right) \\
		\mathrm{s.t.} \ \  &(\ref{1b}), (\ref{1c}).
	\end{align}
\end{subequations}
${{\bar R}_{\text{ZF}}}$ is again a sophisticated function w.r.t. $\mathbf{X}$. Recalling that ${\bf{\bar G}}$ in (P6) exhibits multi-modal behavior, ${\underline{h}_k}$, $\nu$, and hence ${{\bar R}_{\text{ZF}}}$ inherit similar multi-modal characteristics. To address the similar multi-modal optimization problem (P8), the SHADE algorithm, introduced in Section~\ref{sec3}, is again employed. The detailed implementation is provided in \textbf{Algorithm~\ref{alg2}}, with the objective function $f_X$ replaced by $\bar{R}_{\text{ZF}}$.

\subsection{Overall Algorithm, Convergence, and Complexity}
Based on the derived update rules for the digital and analog beamformers, and position optimization, the proposed ZF-based low-complexity algorithm for the FC-THB PASS is summarized in \textbf{Algorithm~\ref{alg4}}. 
\begin{algorithm}[tbp]
	\setlength{\textfloatsep}{0.cm}
	\setlength{\floatsep}{0.cm}
	\small
	\caption{ZF-based Algorithm for FC-THB PASS}
	\renewcommand{\algorithmicrequire}{\textbf{Input}}
	\renewcommand{\algorithmicensure}{\textbf{Output}}
	\label{alg4}
	\begin{algorithmic}[1]
		\REQUIRE System parameters.
		\ENSURE Sub-optimal ${\bf{W}}_{{\rm{BB}}}^{\star}$, ${\bf{W}}_{{\rm{RF}}}^{\star}$, ${\bf{X}}^{\star}$.
		\STATE Initialize ${\bf{W}}_{{\rm{BB}}}$, ${\bf{W}}_{{\rm{RF}}}$, ${\bf{X}}$;
		\STATE Update ${\bf{X}}$ following \textbf{Algorithm~\ref{alg2}} with ${\bar R}_{\text{ZF}} \left( {\bf{X}} \right)$;
		\STATE Calculate ${\bf{W}}_{{\rm{ZF}}}$;
		\REPEAT
		\STATE Update ${\bf{W}}_{{\rm{BB}}}$ by (\ref{hbwbb});
		\STATE Update ${\bf{W}}_{{\rm{RF}}}$ following \textbf{Algorithm~\ref{alg1}};
		\UNTIL a stopping criterion meets.
		\STATE Scale the digital beamformer as ${\bf{W}}_{\rm{BB}} \leftarrow \frac{{\sqrt P }}{{{{\left\| {{{\bf{W}}_{{\rm{RF}}}}{{\bf{W}}_{{\rm{BB}}}}} \right\|}_F}}}{\bf{W}}_{\rm{BB}}$.
	\end{algorithmic}
\end{algorithm}
In this procedure, the position is first optimized to maximize the weighted sum rate achieved by ZF beamforming, after which the ZF beamformer is obtained in closed form. The subsequent decomposition of the transmit beamformer from its ZF counterpart ensures monotonic and bounded convergence. \\
\indent In terms of optimality, since \textbf{Algorithm~\ref{alg4}} adopts also an AO framework, the resulting solution is at most coordinate-wisely optimal for the decomposition. Moreover, as the ZF beamformer itself is sub-optimal with respect to the WSR maximization, the proposed approach trades a slight performance loss for substantially reduced computational complexity. \\
\indent Specifically, the complexity for evaluating ${{\bar R}_{\text{ZF}}}$ is $C_L=\mathcal{O}(K^3+K^2N)$, which is lower than that of computing $f_X$. Hence, the complexity of position optimization is given by ${{i_L}{P_{{\rm{DE}}}}{C_L}}$, where $i_L$ denotes the number of iterations in \textbf{Algorithm~\ref{alg2}}. Compared with \textbf{Algorithm~\ref{alg3}}, \textbf{Algorithm~\ref{alg4}} not only imposes a much lighter computational load on the associated SHADE process, but also requires only a single execution of \textbf{Algorithm~\ref{alg2}} at the beginning.Apart from the position optimization, the computational complexity for updating the two beamformers remains comparable to that of \textbf{Algorithm~\ref{alg3}}.

\section{Numerical results} \label{sec5}
In this section, numerical results are presented to evaluate the performance with respect to the proposed FC architecture for PASS and developed AO and low-complexity algorithms. Furthermore, the comparison with conventional MIMO, widely adopted SC architecture, and the prevailing optimization approach is also provided to show their advantages.

\subsection{Simulation Setup}
The following setup are exploited throughout the simulation unless specified otherwise. Specifically, we assume that there are $K=2$ users randomly located in the $x-y$ plane, i.e., $z_k^u = 0$ within a square region with $D_x = D_y = 10$ m. The dielectric waveguides are uniformly positioned along the $y$-axis, where the inter-waveguide spacing $d_y$ is hence given by $D_y/(M-1)$. The length of each waveguide is set equal to the ground width, i.e., $L = D_x=10$ m. The carrier frequency is set to 30 GHz, and the minimum inter-PA distance is chosen as $\Delta d = \lambda/2 =5$ mm to prevent mutual coupling. The transmit power and the noise power are set to$P = 20$ dBm and $\sigma_k^2 = -90$ dBm, $\forall k$, respectively. The effective refractive index of the dielectric waveguide is set to $n_{\text{eff}} = 1.44$~\cite{neff}. Furthermore, the users are assigned equal serving priority, i.e., $\beta_k = 1/K, \forall k.$ For the proposed \textbf{Algorithm~\ref{alg2}}, the tunable parameters in the SHADE-based position optimization algorithm are configured as $P_{\rm{DE}}=5MN$, $p=20\%$, $G_{\rm{DE}}=100$, and $H_{\rm{DE}}=10$ for the introduced SHADE. For comparison, a grid-search–based position optimization is also employed as a benchmark, where the grid resolution is set to $0.1\lambda=1$ mm. Accordingly, each PA searches over $10^4$ candidate positions along its corresponding waveguide. \\
\indent For the performance comparison, the following benchmark schemes are considered:
\begin{itemize}
	\item \textbf{SC-PASS:} In this benchmark, a prevailing SC architecture~\cite{SCPASS} is adopted. It can be viewed as a special case of the FC PASS with $N_{\text{RF}}=M$ and $\mathbf{W}_{\text{RF}}=\mathbf{I}_M$. Accordingly, the signal received by the $k$-th user can be expressed as $$y_k = {{\bf{h}}_k^H}(\mathbf{X}) {{\bf{W}}_{{\rm{PB}}}}(\mathbf{X}){{\bf{W}}_{{\rm{SC}}}}{\bf{s}} + n_k,$$ where both ${{\bf{W}}_{{\rm{SC}}}}$ and $\mathbf{X}$ are optimized to maximize the weighted sum rate following the principle of the proposed algorithm. The optimization of ${{\bf{W}}_{{\rm{SC}}}}$ follows a similar procedure to that in~\cite{SCPASS}, while the position optimization is conducted using the proposed \textbf{Algorithm~\ref{alg2}}, which achieves a better optimum compared to the PA-wise grid search approach in~\cite{SCPASS}.
	\item \textbf{Massive MIMO:} In this benchmark, a conventional massive MIMO BS is positioned at (0, 5, 5) and equipped with a	uniform planar array parallel to the $y$-axis. A partially connected hybrid beamforming architecture~\cite{Hybrid} is adopted with $MN$ antennas and $M$ RF chains, matching the configuration of the considered SC-PASS. The antennas are assumed to spaced by $\lambda/2$. Under this setup, the signal received by the $k$-th user is given by $$y_k = {{\bf{h}}_k^H}{{\bf{W}}_{{\rm{RF}}}}{{\bf{W}}_{{\rm{BB}}}}{\bf{s}} + n_k,$$ where the $i$-th element in ${\bf{h}}_k$ is calculated by ${\eta \exp \left\{ {j{\textstyle{{2\pi } \over \lambda }}{D_{k,i}}} \right\}}/{{D_{k,i}}}$ with ${D_{k,i}}$ denoting the distance between $i$-th antenna and $k$-th user. The ${{\bf{W}}_{{\rm{RF}}}}$ and ${{\bf{W}}_{{\rm{BB}}}}$ are optimized by the proposed \textbf{Algorithm~\ref{alg3}} with fixed antenna positions.
\end{itemize}

\subsection{Convergence of the Proposed Algorithms}
\begin{figure}[tb]
	\centering
	\includegraphics[scale=0.45]{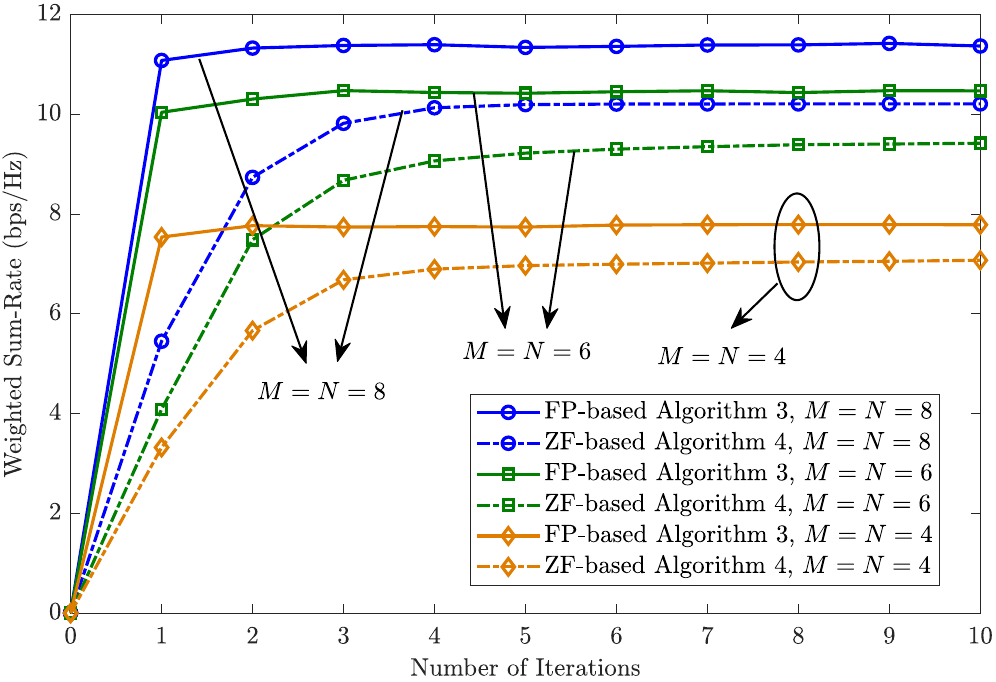}
	\caption{Illustration for the convergence of proposed algorithms}
	\label{convergence}
\end{figure}
The convergence behavior of the proposed algorithms for the FC-THB PASS under different waveguide and PA configurations is illustrated in Fig.~\ref{convergence}. In each setting, both the FP-based \textbf{Algorithm~\ref{alg3}} and the ZF-based \textbf{Algorithm~\ref{alg4}} are evaluated. The number of RF chains is set to half of the number of waveguides, i.e., $N_{\rm RF} = M/2$. The results validate the theoretical analysis and confirm the convergence of the proposed methods. For the FP-based \textbf{Algorithm~\ref{alg3}}, more than 95\% of the final achievable rate is obtained after only one outer iteration, and about 99\% after two iterations. The algorithm typically converges within three iterations. For the ZF-based \textbf{Algorithm~\ref{alg4}}, convergence is achieved within approximately five iterations when the system scale is large, i.e., with a large number of PAs. Moreover, the results are consistent with intuition, i.e., as the numbers of waveguides and PAs increase, the achievable WSR also improves. However, it is noteworthy that the performance gap between the FP-based and ZF-based algorithms widens with increasing $M$ and $N$. This phenomenon arises because, under a fixed total transmit power, enlarging the number of waveguides and PAs per waveguide expands the spatial DoFs of the system. The FP-based \textbf{Algorithm~\ref{alg3}} can flexibly exploit these additional DoFs to enhance the array gain while tolerating controlled interference. In contrast, the ZF-based design imposes strict orthogonality constraints, which increasingly waste spatial resources as $M$ and $N$ grow. Consequently, the performance gap between the optimal and ZF-based schemes widens with the array size.

\subsection{System-level Comparison}
\begin{figure}[tb]
	\centering
	\includegraphics[scale=0.45]{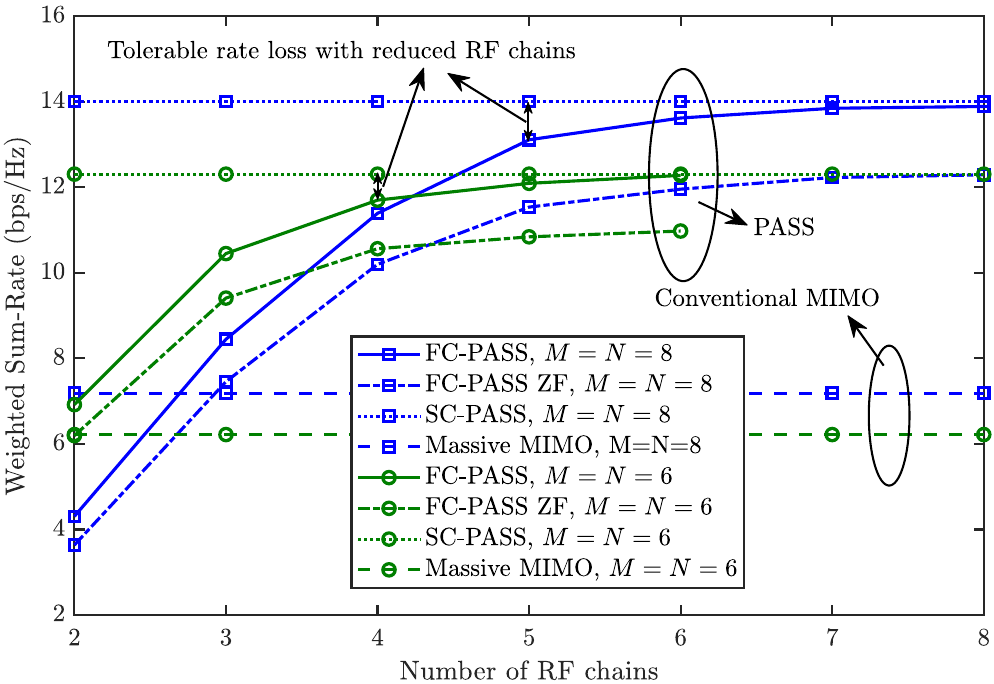}
	\caption{Weighted sum-rate achieved by different architectures versus the number of RF chains}
	\label{sys_com}
\end{figure}
\begin{figure}[tb]
	\centering
	\includegraphics[scale=0.45]{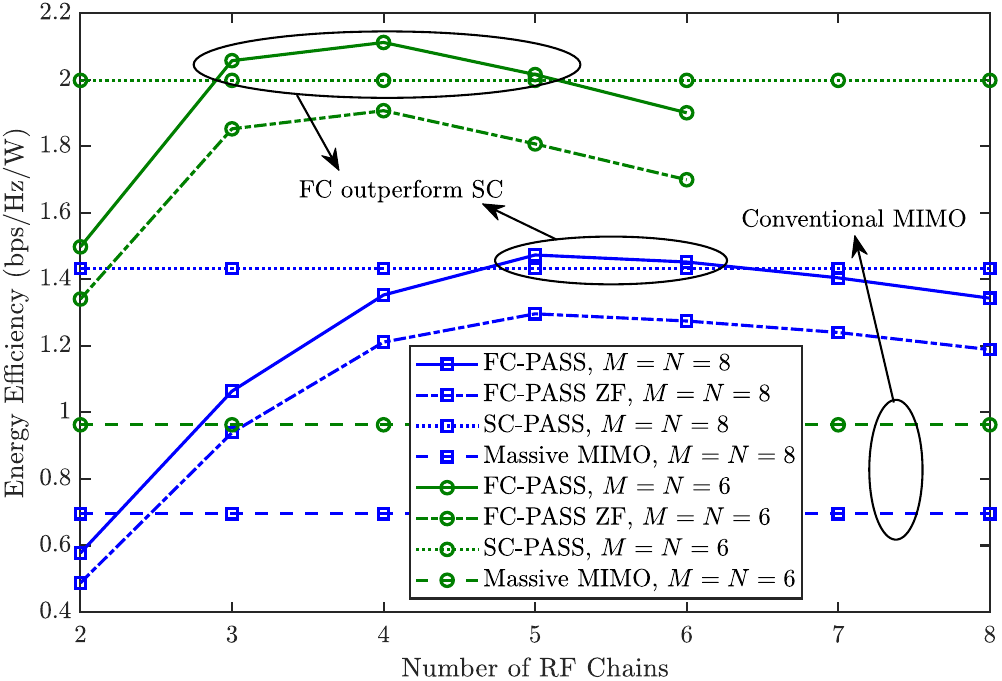}
	\caption{Energy efficiency achieved by different architectures versus the number of RF chains}
	\label{energy}
\end{figure}
\subsubsection{Spectral Efficiency}
In Fig.~\ref{sys_com}, the WSR performance of different architectures and systems is illustrated. The figure also depicts how the results vary with the number of RF chains under two array-size configurations. Specifically, the performance of the proposed FC-THB PASS, obtained using the two developed algorithms, labeled by ``FC-PASS'' and ``FC-PASS ZF'' in Fig.~\ref{sys_com}, is compared with the conventional SC-PASS and the well-established partially connected hybrid beamforming massive MIMO. The results validate the superiority of both the PASS and the proposed FC architecture. In particular, the conventional partially connected hybrid beamforming massive MIMO achieves only about half of the WSR attained by the SC-PASS when the same numbers of antennas and RF chains are employed. \\
\indent When the number of available RF chains is small, both algorithms under the proposed FC-THB PASS yield lower performance than the conventional MIMO system. This is because the overall performance is limited by the RF chain bottleneck, which constrains the rank of the effective channel. As the number of RF chains increases, the performance of the proposed FC-THB PASS gradually approaches that of the SC-PASS. Notably, the performance gap quickly diminishes when the number of RF chains becomes comparable to the number of waveguides, indicating that the FC architecture can achieve near-optimal performance with only a subset of RF chains. For instance, approximately $90\%$ and $95\%$ of the maximum achievable rate can be obtained with five and six RF chains, respectively, when $M=N=8$. \\
\indent As expected, the proposed low-complexity ZF-based algorithm incurs about a 1 dB loss in WSR compared to the sum-rate–driven algorithm; however, this performance gap becomes smaller under lower transmit power or smaller array sizes. These results highlight the flexible trade-off between spectral efficiency and computational complexity enabled by the proposed FC architecture and algorithms. It is also observed that when the number of RF chains is small, systems with more antennas may perform worse than those with fewer antennas. This phenomenon arises because a limited number of RF chains must distribute the same total power among more antennas. Hence, each antenna is allocated with less power but the DoFs in beamforming is restricted, leading to degraded overall performance.

\subsubsection{Energy Efficiency}
The main distinction between the conventional SC-PASS and the proposed FC-PASS architecture lies in the number of phase shifters and the resulting flexible interconnection between RF chains and waveguides. As discussed earlier, in terms of spectral efficiency, the SC-PASS offers more beamforming design DoFs in both amplitude and phase, and thus generally outperforms the FC-PASS. However, when considering power consumption, it becomes important to examine how much energy can be saved by the proposed architecture in exchange for a certain rate loss, i.e., its achievable energy efficiency. Energy efficiency is defined as the ratio between the achievable spectral efficiency and the total power consumption, expressed as
\begin{equation} \label{ee}
	{\eta _{ee}} = \frac{R}{{P + {N_{{\rm{RF}}}}{P_{{\rm{RF}}}} + {N_{{\rm{PS}}}}{P_{{\rm{PS}}}} + {N_{{\rm{PA}}}}{P_{{\rm{PA}}}}}} (\text{bps/Hz/W}),
\end{equation}
where $P$, ${P_{{\rm{RF}}}}$, ${P_{{\rm{PS}}}}$, and ${P_{{\rm{PA}}}}$ denote the transmit power and the power consumed by each RF chain, PS, and power amplifier on PA, respectively. In the simulation, these parameters are set as $P=100$ mW, ${P_{{\rm{PS}}}}=10$ mW, and ${P_{{\rm{PA}}}}=100$ mW~\cite{ManOpt2}. Including the baseband processing cost, the power consumption per RF chain is set to ${P_{{\rm{RF}}}}=400$ mW~\cite{power}. The numbers of components are denoted by ${N_{{\rm{RF}}}}$, ${N_{{\rm{PS}}}}$, and ${N_{{\rm{PA}}}}$, where in the proposed FC-PASS, ${N_{{\rm{PS}}}}={N_{{\rm{RF}}}} \times M$ and ${N_{{\rm{PA}}}} = M \times N$. \\
\indent The energy efficiency results for the proposed FC-PASS, the conventional SC-PASS, and the partially connected hybrid beamforming massive MIMO are illustrated in Fig.~\ref{energy}, under the same settings as those in Fig.~\ref{sys_com}. The results reveal that, similar to the spectral efficiency trends, the PASS achieves approximately twice the energy efficiency of the conventional partially connected hybrid beamforming massive MIMO. Notably, within the PASS family, the proposed FC architecture achieves higher energy efficiency than the SC counterpart when the number of RF chains is moderate, owing to its flexible resource allocation capability. However, as the array size increases, the overall energy efficiency decreases, and the performance advantage of the FC architecture gradually diminishes. \\
\indent These findings verify that the proposed FC-THB PASS achieves a favorable trade-off between spectral efficiency and energy efficiency, offering a flexible and energy-aware design for scalable PASS implementations.

\subsection{Algorithm-level Comparison}
\begin{figure}[tb]
	\centering
	\includegraphics[scale=0.45]{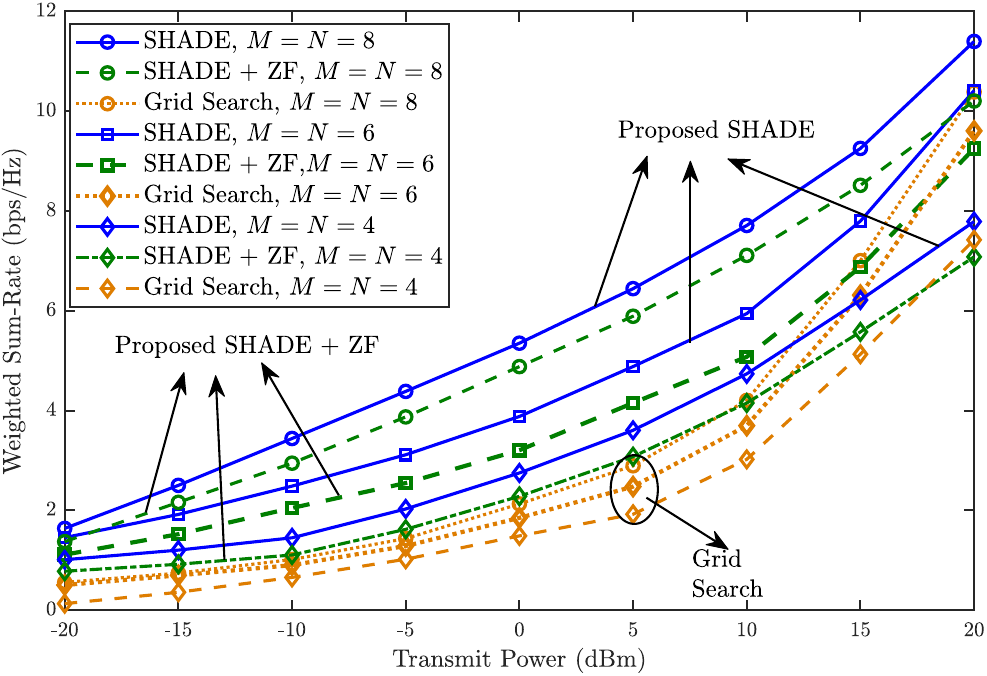}
	\caption{Weighted sum-rate of the proposed FC-THB PASS achieved by different algorithms versus the transmit power}
	\label{alg_com}
\end{figure}
The variation of the WSR achieved by the proposed FC-THB PASS under different algorithms as a function of transmit power is shown in Fig.~\ref{alg_com}. Three methods are compared: the proposed FP-based algorithm (\textbf{Algorithm~\ref{alg3}}), labeled by ``SHADE'', the ZF-based low-complexity algorithm (\textbf{Algorithm~\ref{alg4}}), labeled by ``SHADE +ZF'', and the commonly used PA-wise grid search algorithm. It is observed that the proposed FP-based algorithm consistently achieves the highest spectral efficiency with the proposed SHADE algorithm, while the ZF-based algorithm exhibits less than 1 dB performance loss relative to it due to the adoption of sub-optimal ZF beamforming. The grid search method achieves intermediate performance at high transmit power but performs worst in the low-power regime. \\
\indent Moreover, the performance gap between the grid search and the proposed algorithms widens as the array size increases. This is due to the inherent limitation of the PA-wise approach, which optimizes each PA position individually and therefore only achieves a coordinate-wise local optimum. As the number of PAs grows, the likelihood of converging to a saddle point increases, resulting in suboptimal performance. In contrast, both the proposed methods are based on SHADE, which effectively approaches a near-global optimum for PA position optimization with comparable computational complexity, demonstrating the superiority and efficiency of the proposed algorithms.

\subsection{Impact of serving region}
\begin{figure}[tb]
	\centering
	\includegraphics[scale=0.45]{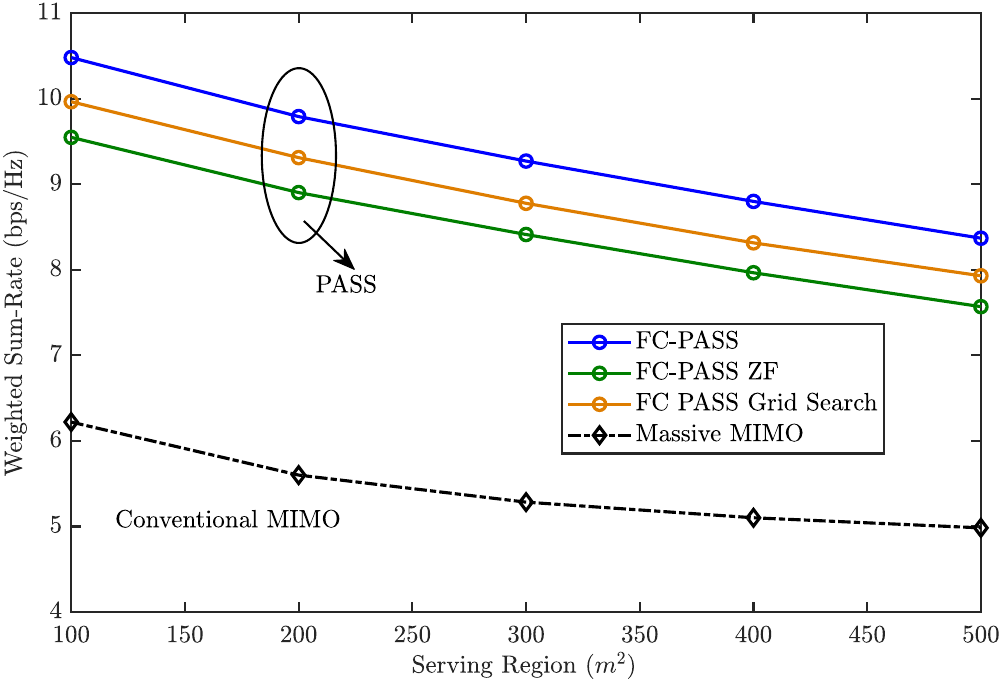}
	\caption{Weighted sum-rate versus the size of serving region}
	\label{region}
\end{figure}
It is also of interest to examine how the achieved WSR varies with the size of the serving region. In this simulation, the region is assumed to be a square, i.e., $D_x=D_y$. The results are illustrated in Fig.~\ref{region}. The FC-PASS is configured with 6 waveguides, each equipped with 6 PAs, and is fed by 3 RF chains. For comparison, the conventional partially connected massive MIMO system employs 36 antennas with 6 RF chains. The transmit power is set to 20 dBm, and all other parameters follow the default settings. \\
\indent Fig.~\ref{region} indicates that the achieved WSR decreases as the serving region expands for both the FC-PASS and the conventional MIMO system. Notably, the rate of reduction is slower for the FC-PASS, reflecting its superior adaptability. For example, when the serving region increases from $100 \text{m}^2$ to $400 \text{m}^2$, i.e., both the length and width are doubled, the achievable WSR of the FC-PASS decreases to approximately $85 \%$ of its original value, whereas it drops to $80 \%$ for the conventional MIMO system. This improvement is attributed to the flexibility of PASS in reconfiguring PA positions to mitigate path loss, highlighting the advantage of the PASS architecture in larger service areas.

\subsection{Robustness to CSI}
\begin{figure}[tb]
	\centering
	\includegraphics[scale=0.45]{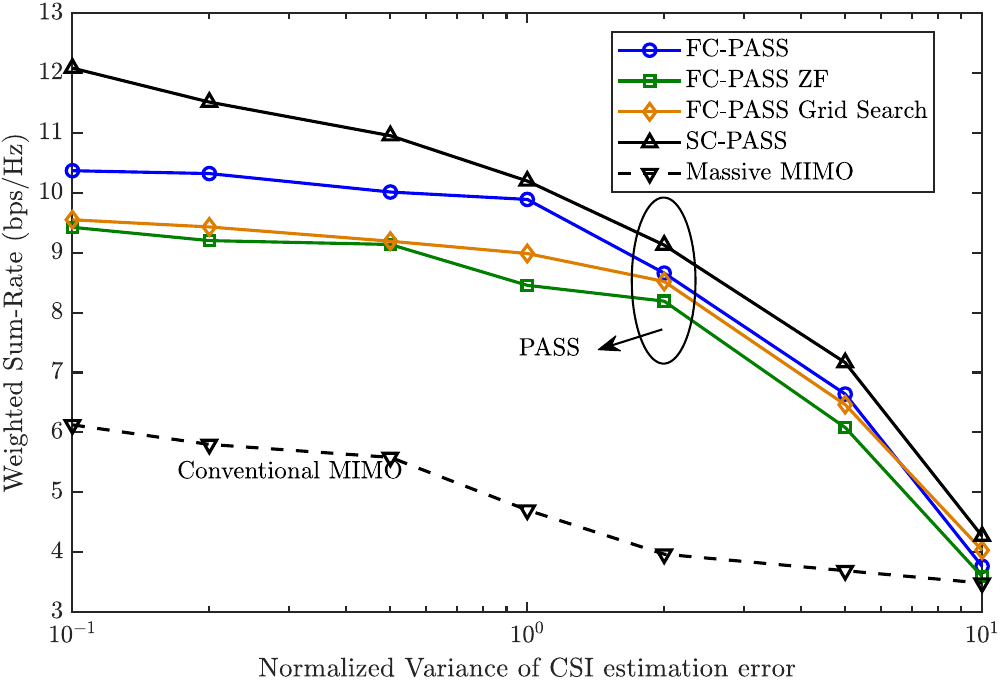}
	\caption{Weighted sum-rate versus the normalized variance of CSI estimation error}
	\label{robustness}
\end{figure}
Fig.~\ref{robustness} investigates the impact of channel estimation errors on the proposed FC-PASS. The simulation adopts the default setting with $M=N=6$. A deterministic model is used to characterize channel estimation errors, where the channel estimated at the BS is perturbed by CSCG noise:
\begin{equation}
	{\mathbf{\hat H}} = \mathbf{H} + \Delta \mathbf{H},
\end{equation}
where ${\mathbf{\hat H}}$ and $\mathbf{H}$ denote the estimated and true channels, respectively. Each element of $\Delta \mathbf{H}$ is assumed to follow $\mathcal{CN}(0, \epsilon)$, where $\epsilon={\textstyle{{\left\| {\bf{H}} \right\|_F} \over {MNK}}}$ representing the normalized variance of the CSI estimation error. The system is optimized based on the estimated channel and evaluated using the actual channel. As shown in Fig.~\ref{robustness}, the achieved spectral efficiency decreases as the CSI error increases, as expected. When the channel estimate becomes less accurate, the performance of the proposed FC-PASS approaches that of the SC-PASS. Nevertheless, even under imperfect CSI, PASS still outperforms the conventional massive MIMO system. However, when the estimation error becomes extremely large, the performance of all architectures and algorithms deteriorates significantly to a similar level.

\section{Conclusion} \label{sec6}
In this paper, we proposed a novel FC architecture for the PASS, which enables flexible connectivity between RF chains and waveguides through a tunable PS network. A joint transmit beamforming and PA position optimization problem was formulated to maximize the WSR under practical power and distance constraints. To solve this challenging non-convex problem, an AO algorithm was developed, complemented by a ZF–based low-complexity alternative. Simulation results verified that the proposed FC-THB PASS can achieve comparable spectral efficiency and superior energy efficiency to the conventional SC architecture with fewer RF chains, while the ZF-based method provides a favorable balance between performance and computational complexity. \\
\indent Since optimizing PA positions still entails considerable computational complexity, it is desirable to develop more lightweight approaches that strike a balance between optimality and complexity. Learning-based or hybrid model-driven methods can also be explored to enable real-time beamforming and position control. Moreover, the current performance of PASS is highly sensitive to the precision of PA positioning due to the rapid phase variation along the waveguide. Therefore, beyond improving the accuracy of PA movement, future work should incorporate practical hardware constraints into the optimization framework to enhance system robustness against hardware impairments.

\appendices
\section{Proof of Lemma~\ref{lemma1}} \label{appa}
Suppose for an arbitrary non-zero ${{\bf{W}}_{{\rm{BB}},0}}$ ,${{\bf{W}}_{{\rm{RF}},0}}$, and ${\bf{X}}_0$, the value of $\bar R$ is ${\bar R}_0$. The SINR in (P1) achieved by $\frac{{\sqrt P }}{{{{\left\| {{{\bf{W}}_{{\rm{RF}},0}}{{\bf{W}}_{{\rm{BB}},0}}} \right\|}_F}}}{\bf{W}}_{{\rm{BB}},0}$ ,${{\bf{W}}_{{\rm{RF}},0}}$, and ${\bf{L}}_0$ is expressed as
\begin{equation}
	{{\gamma _k}} = \frac{{\frac{P}{{\left\| {{{\bf{W}}_0}} \right\|_F^2}}{\rm{Tr}}( {{{\bf{H}}_k}{{\bf{W}}_0}{{\bf{E}}_k}{\bf{W}}_0^H} )}}{{\frac{P}{{\left\| {{{\bf{W}}_0}} \right\|_F^2}}{\rm{Tr}}( {{{\bf{H}}_k}{{\bf{W}}_0}\left( {{{\bf{I}}_K} - {{\bf{E}}_k}} \right){\bf{W}}_0^H} ) + \sigma _k^2}},
\end{equation}
where ${{{\bf{W}}_0}} = {{\bf{W}}_{{\rm{RF}},0}}{\bf{W}}_{{\rm{BB}},0}$. Hence, the value of the objective function is also equal to ${\bar R}_0$. This implies that (\ref{1}) and (\ref{1.5}) can always achieve the same value with a certain relationship in the variables. Moreover, $\frac{{\sqrt P }}{{{{\left\| {{{\bf{W}}_{{\rm{RF}},0}}{{\bf{W}}_{{\rm{BB}},0}}} \right\|}_F}}}{\bf{W}}_{{\rm{BB}},0}$ ,${{\bf{W}}_{{\rm{RF}},0}}$, and ${\bf{x}}_0$ are still feasible for (P1). In a nutshell, (P2) and (P1) and actually equivalent. When (P2) arrives its global optimum at ${\bf{W}}_{{\rm{BB}}}^\star$, ${\bf{W}}_{{\rm{RF}}}^\star$, and ${\bf{X}}^\star$. (P1) is also maximized at $\frac{{\sqrt P }}{{{{\left\| {{{\bf{W}}_{{\rm{RF}}}^\star}{{\bf{W}}_{{\rm{BB}}}^\star}} \right\|}_F}}}{\bf{W}}_{{\rm{BB}}}^\star$, ${\bf{W}}_{{\rm{RF}}}^\star$, and ${\bf{X}}^\star$. 

\section{Example Illustrating the Challenge in PA Position Optimization} \label{challenge}
\begin{figure}[tb]
	\centering
	\includegraphics[scale=0.5]{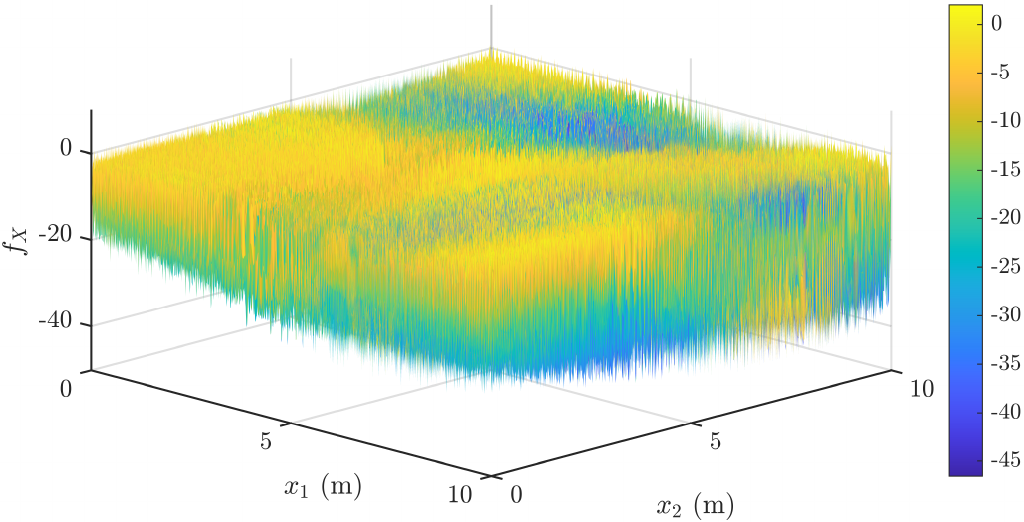}
	\caption{Illustration of the multi-modal objective function in PA position optimization}
	\label{fig_challenge}
\end{figure}
It is observed from (P6) and the expression of ${{\bar g}_{k,m}}\left( {{{\bf{x}}_m}} \right)$ that the phase term in the numerator varies in the wavelength scale. Therefore, in the scenario where the wavelength is small, e.g., 10mm for 30GHz carrier frequency, the change of phase is way too rapidly than the amplitude. This phenomenon leads to the multi-modal characteristic of the PA position optimization. \\
\indent A toy example for this circumstance is illustrated in Fig.~\ref{fig_challenge}. It is assumed that two users are simultaneously served by two waveguides deployed with a PA each, i.e., $M=2,N=1,K=2$ and the two users are on the $xy$-plane with randomly generated positions. Observing from Fig.~\ref{fig_challenge}, in this 10$\text{m}^2$ region, numerous stationary points serve as the local maxima and minima, making the surface of the positions' function behave like a field of golden wheat. This highlights the difficulty in optimizing the PAs' positions. Specifically, the performance gap between local optima can be extremely large due to the significant oscillations in free-space path loss, which result the methods searching stationary solutions fail. In this paper, we adopt the evolution-based algorithm to address this challenge, which is popular in applied math community for solving the multi-modal problem. However, the evolution-based algorithm may introduce a high computational complexity. Further directions will focus on developing a heuristic algorithm based on the specific structure of this problem.

\balance
\bibliographystyle{IEEEtran}
\bibliography{reference/mybib}

\end{document}